\documentclass[11pt,letterpaper]{article}
\usepackage{amssymb,amsmath,amsthm,mathrsfs}
\usepackage{graphicx,epsfig}
\usepackage{wrapfig}
\usepackage{times}
\newtheorem{theorem}{Theorem}
\newtheorem{observation}{Observation}
\newtheorem{proposition}{Proposition}
\newtheorem{corollary}{Corollary}
\newtheorem{lemma}{Lemma}

\newcommand{\fd}{\sc $\mathcal{F}$-Deletion}

\newtheorem{definition}{Definition}

\newcommand{\tw}{{\mathbf{tw}}}

\newcommand{\pmin}{{\sc Min-CMSO}}
\newcommand{\pmax}{{\sc Max-CMSO}}
\newcommand{\pmm}{{\sc Min/Max-CMSO}}

\newcommand{\cO}{\mathcal{O}}

\newcommand{\h}[1]{\end{document}}

\DeclareMathOperator{\operatorClassNP}{\sf NP}
\newcommand{\classNP}{\ensuremath{\operatorClassNP}}

\usepackage{vmargin}

\setmarginsrb{2.54cm}{2.54cm}{2.54cm}{2.54cm}{0pt}{0pt}{0pt}{6mm}

\begin{document}
\title{Bidimensionality and  Geometric Graphs}
\date{}
\author{Fedor V. Fomin\thanks{Department of Informatics, University of Bergen, Norway.
 \newline $~$\hspace{.5cm} 
\texttt{fedor.fomin@ii.uib.no}} 
~~~~Daniel Lokshtanov\thanks{Department of Computer Science and Engineering, University of California, San Diego, USA. 
 \newline $~$\hspace{.5cm} 
\texttt{dlokshtanov@cs.ucsd.edu}} 
~~~~Saket Saurabh\thanks{The Institute of Mathematical Sciences, Chennai, India.
 \newline $~$\hspace{.5cm}
\texttt{saket@imsc.res.in}}
}

\maketitle

\begin{abstract}
\noindent
Bidimensionality theory was introduced by Demaine et al. [{\em JACM 2005}\,] as a framework to obtain algorithmic results for hard problems on minor closed graph classes. The theory has been sucessfully applied to yield subexponential time parameterized algorithms, EPTASs and linear kernels for many problems on families of graphs excluding a fixed graph $H$ as a minor. 
In this paper we use several of the key ideas from Bidimensionality to give a new generic approach to design EPTASs and subexponential time parameterized algorithms for problems on classes of graphs which are not minor closed, but instead exhibit a geometric structure. 
In particular we present EPTASs and subexponential time parameterized algorithms for {\sc Feedback Vertex Set},  {\sc Vertex Cover}, {\sc Connected Vertex Cover}, {\sc Diamond Hitting Set}, on map graphs and unit disk graphs, and for {\sc Cycle Packing} and {\sc Minimum-Vertex Feedback Edge Set} on unit disk graphs. To the best of our knowledge, these results were  previously unknown, with the exception of the EPTAS and a subexponential time parameterized algorithm on unit disk graphs for {\sc Vertex Cover}, which were obtained by Marx  [{\em ESA 2005}\,] and Alber and Fiala [{\em J. Algorithms 2004}\,], respectively. 

 Our results are based on the recent decomposition theorems proved by Fomin et al. in  [{\em SODA 2011}\,] and  novel grid-excluding theorems in unit disc and map graphs without large cliques.  Our algorithms work directly on the input graph and do not require the geometric representations of the input graph.
We also show that our approach can not be extended in its full generality to more general classes of geometric graphs, such as intersection graphs of unit balls in $\mathbb{R}^d$, $d\geq 3$. Specifically, we prove that {\sc Feedback Vertex Set} on unit-ball graphs in $\mathbb{R}^3$ neither admits PTASs unless P=NP, nor subexponential time algorithms unless the Exponential Time Hypothesis fails. Additionally, we show that the decomposition theorems which our approach is based on fail for disk graphs and that therefore any extension of our results to disk graphs would require new algorithmic ideas.
On the other hand, we prove that our EPTASs and subexponential time algorithms for {\sc Vertex Cover} and {\sc Connected Vertex Cover} carry over both to disk graphs and to unit-ball graphs in $\mathbb{R}^d$ for every fixed   $d$. 
\end{abstract}

\section{Introduction}\label{sec:intro}
Algorithms for hard optimization problems on intersection graphs of systems of geometric objects is a well studied area in Computer Science motivated by various applications in wireless networks \cite{KuhnWZ08}, computational biology \cite{XuB06} and map labeling \cite{AgarwalKS98}. While most problems remain \classNP{}-complete even when restricted to such classes, the restriction of a problem to a geometric class is usually much more tractable algorithmically than the unrestricted problem. For example, for planar graphs and more generally minor closed families of graphs, the Bidimensionality theory of Demaine et al.~\cite{DFHT05} simultaneously demonstrates the tractability of most natural problems with respect to subexponential time parameterized algorithms~\cite{DFHT05}, approximation~\cite{DH05,FLRSsoda2011} and kernelization~\cite{F.V.Fomin:2010oq}. For related ``geometric'' classes of graphs that are not closed under taking minors, the picture is considerably more heterogenous and the situation is less understood. The objective of this article is to take a step towards clearing the picture for geometric graph classes. 


Most of the known approximation schemes that have been obtained for graph problems on geometric graph classes use a variation of the well-known {\em shifting} technique introduced in the classical works of Baker \cite{Baker94} and of Hochbaum and Maass \cite{Hochbaum:1985:}. Hunt et al. \cite{huntetal98} used the shifting technique to give polynomial time approximation schemes (PTASs) for a number of problems such as \textsc{Maximum Independent Set} and \textsc{Minimum Dominating Set} on unit disk graphs and $\lambda$-precision disk graphs. Independently,  Erlebach et al. \cite{Erlebach:2005} and Chan~\cite{Chan:2003} generalized the shifting technique and gave PTASs for \textsc{Maximum Independent Set} and \textsc{Minimum Vertex Cover} on disk graphs and on intersection graphs of fat objects. Marx in \cite{marx-approx} obtained an efficient polynomial time approximation schemes (EPTAS) for \textsc{Minimum Vertex Cover} on unit disk graphs. Chen in \cite{Chen01} and Demaine et al. \cite{DemaineFHT05talg} used similar approaches to obtain a PTAS for \textsc{Maximum Independent Set} and \textsc{Minimum $r$-Dominating Set} on map graphs. One of the known limitations of the shifting technique is that it generally only applies to local problems such as {\sc Vertex Cover} and variants of {\sc Dominating Set}, and fails for non-local problems such as {\sc Feedback Vertex Set} and {\sc Cycle Packing}. For problems on planar and $H$-minor-free graphs, Bidimensionality is able to handle non-locality by applying {\em treewidth} based decomposition. It is tempting to ask whether treewidth based decomposition can be useful for other graph classes as well.

In this article we use key ideas from Bidimensionality and design a general approach that can be used to give EPTASs and subexponential time parameterized algorithms for many problems on map graphs and for unit disk graphs, and in some cases on even more general geometric classes of graphs. We present EPTASs and subexponential time parameterized algorithms for {\sc Feedback Vertex Set},  {\sc Vertex Cover}, {\sc Connected Vertex Cover}, {\sc Diamond Hitting Set}, {\sc Minimum-Vertex Feedback Edge Set} on map   and unit disk graphs, and for {\sc Cycle Packing} and {\sc Minimum-Vertex Feedback Edge Set} on unit disk graphs. Our approach is based on the concept of {\em truly sublinear treewidth}, recently introduced by the authors in~\cite{FLRSsoda2011} as a tool to give EPTASs for bidimensional problems on minor closed graph classes. Roughly speaking, a graph class has truly sublinear treewidth if adding $k$ vertices to any graph in the class such that the resulting graph is in the class as well,   increases its treewidth by   $O(k^\epsilon)$ for $\epsilon < 1$. The techniques in~\cite{FLRSsoda2011} can not be applied directly to map graphs and unit disk graphs, because both graph classes contain arbitrarily large cliques and hence do not have truly sublinear treewidth. We overcome this obstacle by showing that cliques are the only pathological case. Namely, we prove that map graphs and unit disk graphs that exclude large clique subgraphs have truly sublinear treewidth. Our EPTASs work in two steps, first we ``clean" the input graph for large cliques, and then we apply the decomposition theorems from~\cite{FLRSsoda2011}. 

The initial application of Bidimensionality was in the design of subexponential parameterized algorithms on planar, and more generally, on $H$-minor-free graphs \cite{DFHT05}. Demaine et al.~\cite{DemaineFHT05talg} used Bidimensionality to obtain subexponential parameterized algorithms for {\sc Dominating Set}, and more generally, for  \textsc{$(k,r)$-Center} on map graphs. We show that  after ``cleaning" unit disk and map graphs from large cliques, it is possible to use Bidimensionality to solve many  parameterized problems in subexponential time on these classes of graphs.  To the best of our knowledge, prior to our work parameterized subexponential algorithms on unit disk graphs were known only for \textsc{Vertex Cover}  \cite{AF04}. The important ingredient of our algorithms are the analogues of excluding grid theorems of  Robertson et al.  for planar graphs \cite{RobertsonST94} and of of Demaine and Hajiaghayi for $H$-minor free graphs \cite{Demaine:2008dq}. We show that the treewidth  of every unit disc or map graph excluding  a clique of constant size as a subgraph  and excluding   a $k\times k$ grid as a minor, is $\cO(k)$.
%
%

Our algorithms do not require  geometric representations of the input graphs. Since recognition of unit disk graphs is \classNP{}-hard~\cite{ClarkCJ90} and the exponent of the polynomial bounding the running time of  map graph recognition algorithm  is about 120~\cite{Thorup98a}, the robustness of our algorithms is a serious advantage. 

We explore to which degree our approach can be lifted to other classes of graphs. Our investigations show that it is unlikely that the full power of our approach can be generalized to disk graphs or to unit ball graphs in $\mathbb{R}^d$ --- intersection graphs of unit-balls in $\mathbb{R}^d$, $d\geq 3$. Specifically we prove that {\sc Feedback Vertex Set} on unit-ball graphs in $\mathbb{R}^3$ neither admits a PTASs unless P=NP, nor a subexponential time algorithm unless the Exponential Time Hypothesis fails. Furthermore we show that disk graphs which exclude the clique on four vertices as a subgraph do not have truly sublinear treewidth. On the other hand, an adaptation of our techniques yields EPTASs and subexponential time parameterized algorithms for {\sc Vertex Cover} and {\sc Connected Vertex Cover} both on disk graphs and on unit disk graphs in $\mathbb{R}^d$ for every fixed integer dimension $d$.

A natural question is whether our results can be extended to handle larger classes of problems on map graphs and unit disk graphs. It appears that the main obstacle to generalizaing our approach is to design more general clique cleaning procedures.
In particular, Marx~\cite{marx-approx} showed that {\sc Dominating Set} and {\sc Independent Set} are W[1]-hard even on unit disk graphs. This means that the two problems neither admit EPTASs nor FPT algorithms unless FPT=W[1], a complexity collapse considered very unlikely. 
While no clique cleaning procedure is known for these two problems, a simple modification of our techniques show that both problems admit both EPTASs and subexponential time algorithms on map graphs and unit disk graphs excluding large cliques as subgraphs. Thus it seems that we are able to handle exactly the problems for which cliques can be removed efficiently.

\section{Definitions and Notations}

In this section we give various definitions which we make use of in the paper. Let~$G$ be a graph with vertex set $V(G)$ and edge set $E(G)$. A graph~$G'$ is a  \emph{subgraph} of~$G$ if~$V(G') \subseteq V(G)$ and~$E(G') \subseteq E(G)$. 
The subgraph~$G'$ is called an \emph{induced subgraph} of~$G$ if~$E(G')
 = \{ uv \in E(G) \mid u,v \in V(G')\}$, in this case, $G'$~is also called the subgraph \emph{induced by~$V(G')$} and denoted by~$G[V(G')]$. For a vertex set $S$, by $G \setminus S$ we denote $G[V(G) \setminus S]$. A graph class ${\cal G}$ is {\em hereditary} if for any graph $G \in {\cal G}$ all induced subgraphs of $G$ are in ${\cal G}$. 
  By $N(u)$ we denote (open) neighborhood of $u$, that is, the set of all vertices adjacent to $u$. Similarly, by $N[u]=N(u) \cup \{u\}$ we define the closed neighborhood.  The degree of a vertex $v$ in $G$ is $|N_G(v)|$. We denote by $\Delta(G)$ the maximum vertex degree in $G$. For a subset $D \subseteq V(G)$, we define $N[D]=\cup_{v\in D} N[v]$ and $N(D) = N[D] \setminus D$. 
%
  Given an edge  $e=xy$ of a graph $G$, the graph  $G/e$ is obtained from  $G$ by contracting the edge $e$. That is, the endpoints $x$  and $y$ are replaced by a new vertex $v_{xy}$
which  is  adjacent to the old neighbors of $x$ and $y$ (except from $x$ and $y$).  A graph $H$ obtained by a sequence of edge-contractions is said to be a \emph{contraction} of $G$.  
 A graph $H$ is a {\em minor} of a graph $G$ if $H$ is the contraction of some subgraph of $G$ and we denote it by $H\leq_{m} G$. We also 
use the following equivalent characterization of minors.
\begin{proposition}[\cite{Diestel2005}]\label{prop:branchsets}
 A graph $H$ is a minor of $G$ if and only if there is a map 
$\phi:V(H)\rightarrow 2^{V(G)}$  such that for every vertex $v\in V(H)$,  
$G[\phi(v)]$ is connected, for every pair of vertices $v,u\in V(H)$,  
$\phi(u)\cap \phi(v)=\emptyset$, and for every edge $uv\in E(H)$, there is an 
edge $u'v'\in E(G)$ such that $u'\in \phi(u)$ and $v'\in \phi(v)$. 
\end{proposition}
 Let $G,H$ be two graphs. A subgraph $G'$ of $G$ 
is said to be a \emph{minor-model} of $H$ in $G$ if $G'$ contains $H$
as a minor. 
 The
\emph{$(r\times r)$-grid} is the Cartesian product of two paths of
lengths $r-1$. 

\medskip
\noindent\textbf{\bf Treewidth}
A \emph{tree decomposition} of a graph $G$ is a pair $(\mathcal{X},T)$, where $T$
is a tree and ${\cal X}=\{X_{i} \mid i\in V(T)\}$ is a collection of subsets
of $V$ such that the following conditions are satisfied.
\begin{enumerate}\setlength\itemsep{-1.2mm}
\item $\bigcup_{i \in V(T)} X_{i} = V(G)$.
\item For each edge $xy \in E(G)$, $\{x,y\}\subseteq X_i$ for some $i\in V(T)$.
\item For each $x\in V(G)$ the set $\{ i \mid x \in X_{i} \}$ induces a connected subtree of $T$.
\end{enumerate}
Each $X_i$ is called the bag of a tree decomposition. 
The \emph{width} of the  tree decomposition is $\max_{i \in V(T)}\,|X_{i}| - 1$. The \emph{treewidth} of a
graph $G$, ${\bf tw}(G)$, is the minimum width over all tree decompositions of $G$.
%


\medskip
\noindent\textbf{Plane, unit disk and map graphs.}
In this paper we  use the expression \emph{plane graph} for any planar graph drawn in the Euclidean plane $\mathbb{R}^2$ without any edge crossing. We do not distinguish between a vertex of a plane graph and the point of  $\mathbb{R}^2$ used in the drawing to represent the vertex or between an edge and the {curve} representing it. We also consider plane graph $G$ as the union of the points corresponding to its vertices and edges.  We call by {\em face} of $G$ any connected component of $\mathbb{R}^2 \setminus (E(G)\cup V(G))$.  The \emph{boundary} of a face is the set of edges incident to it. If the boundary of a face $f$ forms a cycle then we call it a {\em cyclic face}.  
A \emph{disk graph} is the intersection graph of a family of (closed) disks in  $\mathbb{R}^2$. A \emph{unit disk graph} is the intersection graph of a family of unit disks in $\mathbb{R}^2$. The notion of a map graph is due to 
Chen et al.  
\cite{ChenGP98}.  A {\em map} $\cal M$ is a pair  $(\mathscr{E},\omega)$, where $\mathscr{E}$  is a plane graph and each connected component of $\mathscr{E}$ is  biconnected, and $\omega$ is a function that maps each face $f$ of $\mathscr{E}$ to $0$ or $1$ in a way that whenever $\omega(f)=1$, $f$ is a closed face.  A face $f$ of $\mathscr{E}$ is called {\em nation} if $\omega(f)=1$, {\em lake} otherwise. The graph associated with $\cal M$ is a simple graph $G$, where $V(G)$ consists of the nations on $\cal M$ and $E(G)$ consists of all $f_1f_2$ such that $f_1$ and $f_2$ are adjacent (that is shares at least one vertex). We call $G$ a {\em map graph}. By $N(\mathscr{E})$ we denote  the set of nations of $\mathscr{E}$.

\medskip
\noindent\textbf{Counting Monadic Second Order Logic.}
\label{countmsop}
The syntax of MSO of graphs includes the logical connectives $\vee$, $\land$, $\neg$, $\Leftrightarrow $,  $\Rightarrow$, variables for vertices, edges, set of vertices and set of edges, the quantifiers $\forall$, $\exists$ that can be applied to these variables, and the following five binary relations: 
\begin{enumerate}\setlength\itemsep{-1.2mm}
\item $u\in U$ where $u$ is a vertex variable and $U$ is a vertex set variable.
\item $d \in D$ where $d$ is an edge variable and $D$ is an edge set variable.
\item $\mathbf{inc}(d,u)$, where $d$ is an edge variable,  $u$ is a vertex variable, and the interpretation is that the edge $d$ is incident on the vertex $u$.
\item $\mathbf{adj}(u,v)$, where  $u$ and $v$ are vertex variables, and the interpretation is that $u$ and $v$ are adjacent.
\item Equality of variables, $=$, representing vertices, edges, set of vertices and set of edges.
\end{enumerate}
 {\em Counting monadic second-order logic} (CMSO)  is  {   monadic second-order logic} (MSO) additionally equipped with an atomic formula $\mathbf{card}_{n,p}(U)$ for testing whether the cardinality of a set $U$ is congruent to $n$ modulo $p$, where $n$ and $p$ are integers independent of the input graph such that $0\leq n<p$ and $p\geq 2$.
We refer to ~\cite{ArnborgLS91,Courcelle90,Courcelle97} for a detailed introduction to CMSO.
\pmin{} and \pmax{} problems are graph optimization problems where the objective is to find a maximum or minimum sized vertex or edge set satisfying a CMSO-expressible property. 
In particular, in a \pmm{} graph problem $\Pi$ we are given a graph $G$ as input. The objective is to find a minimum/maximum cardinality vertex/edge set $S$ such that the CMSO-expressible predicate $P_\Pi(G,S)$ is satisfied.
%
%

\medskip
\noindent\textbf{Bidimensionality and Separability.}
Our results concern graph optimization problems where the objective is to find a vertex or edge set that satisfies a feasibility constraint and maximizes or minimizes a problem-specific objective function. For a problem $\Pi$ and vertex (edge) set $S$ let $\phi_\Pi(G,S)$ be the feasibility constraint returning {\bf true} if $S$ is feasible and {\bf false} otherwise. Let $\kappa_\Pi(G,S)$ be the objective function. 
In most cases, $\kappa_\Pi(G,S)$ will return $|S|$. We will only consider problems where every instance has at least one feasible solution. Let ${\cal U}$ be the set of all graphs. 
For a graph optimization problem $\Pi$ let $\pi : {\cal U} \rightarrow \mathbb{N}$ be a function returning the objective function value of the optimal solution of $\Pi$ on $G$. We say that a problem $\Pi$ is {\em minor-closed} if $\pi(H) \leq \pi(G)$ whenever $H$ is a minor of $G$.
 We now define bidimensional problems.

\begin{definition}[\cite{DFHT05}]
  A graph optimization problem $\Pi$ is minor-bidimensional if 
  $\Pi$ is minor-closed and 
  there is   $\delta > 0$ such that $\pi(R) \geq \delta r^2$ for the $(r \times r)$-grid $R$. In other words, the value of the solution on $R$ should be at least   $ \delta r^2$.
\end{definition}

Demaine and Hajiaghayi~\cite{DH05} define the  \emph{separation} property for problems, and show how separability together with bidimensionality is useful to obtain EPTASs on $H$-minor-free graphs. In our setting a slightly weaker notion of separability is sufficient. In particular the following definition is a reformulation of the requirement $3$ of the definition of separability in~\cite{DH05} and similar to the definition used in~\cite{F.V.Fomin:2010oq} to obtain kernels for bidimensional problems. 

\begin{definition} \label{def:sep}
A minor-bidimensional problem $\Pi$ has the \emph{separation} property if given any graph $G$ and a partition of $V(G)$ into $L \uplus S \uplus R$ such that $N(L) \subseteq S$ and $N(R) \subseteq S$, and given an optimal solution $OPT$ to $G$, $\pi(G[L]) \leq \kappa_\Pi(G[L], OPT \cap L) + 
\mathcal{O}(|S|)$ and $\pi(G[R]) \leq \kappa_\Pi(G[R], OPT \cap R) + \mathcal{O}(|S|)$.
\end{definition}

In Definition~\ref{def:sep}  we slightly misused notation. Specifically, in the case that $OPT$ is an edge set we should not be considering $OPT \cap R$ and $OPT \cap L$ but $OPT \cap E(G[R])$ and $OPT \cap E(G[L])$ respectively.

\medskip
\noindent\textbf{Reducibility, $\eta$-Transversability and Graph Classes with Truly Sublinear Treewidth.}
We now define three of the central notions of this article.
 
\begin{definition}
A graph optimization problem $\Pi$ with objective function $\kappa_\Pi$ is called {\em reducible} if there exist a \pmm{} problem $\Pi'$ and a function $f : \mathbb{N} \rightarrow \mathbb{N}$ such that 
\begin{enumerate}\setlength\itemsep{-1.2mm}
 \item there is a polynomial time algorithm that given $G$ and $X \subseteq V(G)$ outputs $G'$ such that $\pi'(G') = \pi(G) \pm \cO(|X|)$ and $\tw(G') \leq f(\tw(G \setminus X))$,
 \item there is a polynomial time algorithm that given $G$ and $X \subseteq V(G)$, $G'$ and a vertex (edge) set $S'$ such that $P_{\Pi'}(G',S')$ holds, outputs $S$ such that $\phi_\Pi(G,S)=\textbf{true}$   and $\kappa_\Pi(G,S)=|S'| \pm \cO(|X|)$.
\end{enumerate}
\end{definition}

\begin{definition}
A graph optimization problem $\Pi$ is called $\eta$-transversable if there is a polynomial time algorithm that given a graph $G$ outputs a set $X$ of size $\cO(\pi(G))$ such that $\tw(G \setminus X) \leq \eta$.
\end{definition}


\begin{definition}  
Graph class ${\cal G}$ has {\em truly sublinear treewidth with parameter}  $\lambda$, $ 0<\lambda<1$, if for every  $\eta>0$, there exists  
$\beta>0$ such that for any graph $G \in {\cal G}$ and $X \subseteq V(G)$ the condition  $\tw(G \setminus X) \leq \eta$ yields that $\tw(G) \leq \eta + \beta|X|^\lambda$. \end{definition}

\noindent 
Throughout this paper $K_t$ denotes a complete graph on $t$ vertices and we say that a graph $G$ is {\em $K_t$-free} if $G$ does not contain $K_t$ as 
an induced subgraph. 

\section{Structure of $K_t$-free Geometric Graphs}
In this section we show that unit disk graphs and map graphs with bounded 
maximum clique size has truly sublinear treewidth.  

\subsection{Structure of $K_t$-free Unit Disk Graphs.}
\label{sec:udgoofbd}
We start with the following well known observation (see \cite[Lemma~3.2]{MaratheBHR95}) about $K_t$-free unit disk graphs.
\begin{observation}
\label{obs:timpliesd}
If a unit disk graph $G$ is $K_t$-free then  $\Delta(G)\leq 6t$. 
\end{observation}
Observation~\ref{obs:timpliesd} allows us to prove theorems on unit disk graphs of bounded maximum degree and then use these results for $K_t$-free graphs. 

Let $G$ be a unit disk graph generated by ${\cal B}=\{B_1,\ldots,B_n\}$ and $\Delta(G)=\Delta$. 
We will associate an auxiliary planar graph $P_G$ with $G$ such that the treewidth of these 
two graphs is linearly related.  Let $P_I$ be a planar graph defined as follows. Consider the 
embedding (drawing) of the unit disks ${\cal B}=\{B_1,\ldots,B_n\}$ in the plane. Let $\cal P$ be the set of points in 
the plane such that each point in the set is on the boundary of at least two disks. Essentially, this is the set of points at 
unit distance to centers of at least two disks.  
We place a vertex at each point in $\cal P$ and regard the curve between 
a pair of vertices as an edge,  then the 
embedding of unit disks ${\cal B}=\{B_1,\ldots,B_n\}$ in the plane gives rise to the drawing $P_I$ of a planar multigraph. 
Furthermore let $D_I$ be the planar dual of $P_I$; it is well known that $D_I$ is also planar. 


Next we define a notion of {\em region} which is essential for the definition of $P_G$. 
We call a face $\cal R$ of the plane graph $P_I$ a  {\em region}, if there exists a nonempty subset 
${\cal B}'\subseteq {\cal B}$ of unit disks  such that every point in $\cal R$ is an interior point of each disk in ${\cal B}'$. 
Hence with every region $\cal R$ we can associate a set of unit disks.  
Since the vertices of $G$ correspond to disks of ${\cal B}$, we can associate a subset of 
vertices of $G$, say ${\cal V}({\cal R})$, to a region $\cal R$. 
We remark that there could be two regions ${\cal R}_1$ and ${\cal R}_2$ with ${\cal V}({\cal R}_1)={\cal V}({\cal R}_2)$. 
Now we are ready to define the graph $P_G$.   
Let ${\cal R}_1,\ldots,{\cal R}_p$ be  the regions of $P_I$. These are faces in $P_I$ and hence in the dual graph $D_I$ we have 
vertices corresponding to them. That is, in $D_I$  for every region ${\cal R}_i$ we have a vertex  
$v({\cal R}_i)$.  We define 
\[P_G:=D_I[ \{v({\cal R}_i) \mid 1 \leq i \leq p\}].\]
Thus, $P_G$ is an induced subgraph of $D_I$ obtained by removing non-region vertices. Figure~\ref{unitdg} illustrates  the construction of graphs $P_I$ and $P_G$ from unit disks drawing.
%
Next we prove some properties of $P_G$. 
\begin{figure}
  \begin{center} 
    \includegraphics[scale=.3]{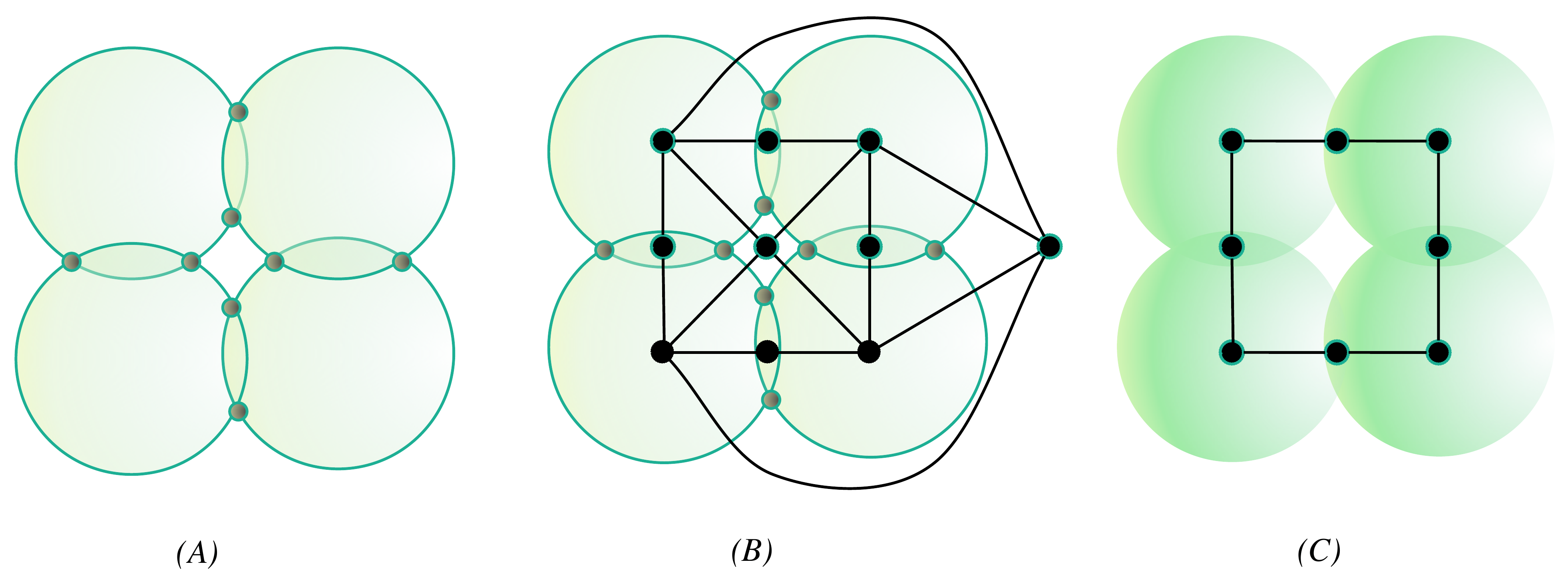}
   \end{center}
   \caption{{{(A) Drawing of a planar (multi)graph $P_{I}$ formed by the drawing of four disks; (B) The dual graph $D_I$ of $P_I$; (C) The graph $P_G$.}}}
\label{unitdg}
 \end{figure}

\begin{lemma}
\label{lem:pgisplanar}
 Let $G$ be a unit disk graph of maximum degree $\Delta$. 
 Then  $P_G$ is a planar graph and every vertex $v\in V(G)$ is a part of  at most $3(\Delta^2+\Delta)$ regions. 
\end{lemma}

\begin{proof}
The graph $P_G$ is a subgraph of $D_I$, 
the planar dual of $P_I$, and hence it is also planar. 
Let $v\in V(G)$ be a vertex. We consider the 
embedding (drawing) of unit disks corresponding to the vertices of the closed neighborhood 
 $ N_G[v]$ in the plane. Then  $|N_G[v]|\leq \Delta +1$.  Let $\cal L$ be the set of the points in 
the plane such that each point in the set is on the boundary of at least two disks with distinct center points. This is
the induced subgraph of $P_I$ formed by the intersection points of the boundaries of disks from  $N_G[v]$. 
  Since every two circles with distinct center points intersect in at most  two points,
   we have that 
$|{\cal L}|\leq 2 {|N_G[v]|\choose 2}\leq 2 {\Delta+1 \choose 2}= \Delta^2+\Delta$. 
Consider the planar graph $P_I[{\cal L}]$, which is a subgraph of $P_I$ induced by ${\cal L}$. Observe that  $v$ can only  
 be a part of regions defined by faces of $P_I[{\cal L}]$. To obtain an upper bound on the number of faces  of $P_I[{\cal L}]$, 
we first obtain an upper bound on the number of edges of $P_I[{\cal L}]$. First observe that between any pair of vertices in 
$P_I[{\cal L}]$ there can at most be two edges and there are at most $|{\cal L}|$ pairs that have two edges between them. 
It is well known that  a planar graph on $n$ vertices without any parallel edges has at most $3n-6$ edges. Thus, the number of 
edges in  $P_I[{\cal L}]$  is at most $3|{\cal L}|-6+|{\cal L}|=4|{\cal L}|-6$. Now by  Euler's formula,  
 the number of faces   in 
$P_I[{\cal L}]$  is at most  
 \[2+|E(P_I[{\cal L}])|-|{\cal L}|\leq 2+ 3 |{\cal L}|-6\leq 2+ 3(\Delta^2+\Delta)-6\leq 3(\Delta^2+\Delta).\] 
 Thus $v$   is a part of at most $3(\Delta^2+\Delta)$ regions. We complete the proof by making a remark that 
the Euler's formula also holds for graphs with multiple edges. 
 \end{proof}

\begin{lemma} 
\label{lem:twofg}
 Let $G$ be a unit disk graph. Then $\tw(G)\leq (\Delta(G)+1) \cdot (\tw(P_G)+1)-1$.
\end{lemma}
\begin{proof}
Let $\Delta(G)=\Delta$ and  $(\mathcal{X}',T)$ be a tree decomposition of   $P_G$ of width $\tw(P_G)$. We build  a tree 
decomposition $(\mathcal{X},T)$ of  $G$ from  the tree-decomposition  $(\mathcal{X}',T)$ 
of $P_G$.  Let $X_i'$ be the subset of $V(P_G)$ associated with the node $i$ of $T$.  
We define $X_i:=\bigcup_{v({\cal R})\in X_i'}{\cal V}({\cal R})$. Recall that ${\cal V}({\cal R})$ is a subset of vertices in 
$V(G)$ characterizing $\cal R$. 
This concludes the description of a decomposition for $G$.  
 Observe that the set 
${\cal V}({\cal R})$ is contained in $N_G[w]$ for every $w\in {\cal V}({\cal R})$ and hence the size of each of 
them is bounded above by $\Delta+1$. Hence the size of each of $X_i$ is at most $(\Delta +1)\cdot |X_i'|$. 
This implies that the size of every  bag $X_i$ is at most $(\Delta+1) (\tw(P_G)+1)$. 

Now we show that this is indeed a tree-decomposition for $G$ by proving that it  satisfies the three properties of  a tree decomposition.  
By construction,  every vertex of $V(G)$ is contained in some $X_i$. To show that for every edge 
$uv \in E(G)$ there is a node $i$ such that $u,v \in X_i$,  we argue as follows. If there is an edge between $u$ and $v$ in 
$G$ then  unit disks corresponding to these vertices intersect and hence there is a region $\cal R$ which is 
completely contained inside this intersection. This implies that $u, v \in {\cal V}({\cal R}) $. For node $i$ such that $v({\cal R})$ is 
contained inside $X_i'$, we have that the corresponding bag $X_i$ contains $u$ and $v$.  To conclude we  need 
to show that for each $v \in V(G)$ the set $Z=\{ i \mid x \in X_{i} \}$ induces a   subtree of $T$. Observe that 
$v$ appears in all the  bags corresponding to node $i$ such that $X_i'$ contains a vertex corresponding to a region   which 
$v$ is a part of. This implies that all these regions are inside the unit disk corresponding to $v$. Hence the graph induced by 
vertices corresponding to these regions  is connected.  Thus the set $Z$  induces a subtree of $T$. 
\end{proof}

We now show a linear excluded grid theorem for unit disk graphs of bounded degree.

\begin{lemma}\label{lem:udglingrid}
Any unit disk graph $G$ with maximum degree $\Delta$ contains a $\frac{\tw(G)}{100\Delta^3} \times \frac{\tw(G)}{100\Delta^3}$ grid as a minor.
\end{lemma}
\begin{proof}
Let $G$ be a unit disk graph of maximum degree $\Delta$, and define $P_G$ as above.  Since $P_G$ is planar, by the excluded grid theorem for planar graphs~\cite{RobertsonST94}, $P_G$ contains a $t \times t$ grid as a minor where $t=\frac{\tw(P_G)}{6}$. By Proposition~\ref{prop:branchsets},  we know that there is a minor model of this grid, say 
$\{S[i,j] : 1\leq i,j \leq t\}$.   
We know that for every $i$,$j$, $P_G[S[i,j]]$ is connected, the sets $S[i,j]$ are pairwise disjoint and finally for every $i$,$j$,$i'$,$j'$ such that $|i-i'|+|j-j'|=1$ there is an edge in $P_G$ with one endpoint in $S[i,j]$ and the other in $S[i',j']$. For every $i$,$j$, we  build $S'[i,j]$ from $S[i,j]$ by replacing every vertex $v({\cal R}) \in S[i,j]$ by ${\cal V}({\cal R})$ and removing duplicates. We set $\Delta' = 3(\Delta^2+\Delta)$, and observe that for any vertex $v$ in $G$, Lemma~\ref{lem:pgisplanar} implies that there are at most $\Delta'$ sets $S'[i,j]$ that contain $v$. 

We say that an integer pair $(i,j)$ is {\em internal} if $\Delta' \leq i \leq t-\Delta'$ and $\Delta' \leq j \leq t-\Delta'$. We prove that for any two internal pairs $(i,j)$ and $(i',j')$ such that $|i-i'|+|j-j'|>\Delta'$ the sets $S'[i,j]$ and $S'[i',j']$ are disjoint. To obtain a contradiction 
assume  that both sets contain a vertex $v$ in $G$. Let $X_v$ be the set of vertices $v({\cal R})$ such that $v \in {\cal V}({\cal R})$. We will show that $|X_v|>\Delta'$ which contradicts that $v$ is part of at most $\Delta'$ regions. On one hand, $P_G[X_v]$ is connected. On the other hand, both $S[i,j] \cap X_v$ and $S[i',j'] \cap X_v$ are non-empty. But any path in $P_G$ between a vertex in $S[i,j]$ and a vertex in $S[i',j']$ must pass through at least $\Delta'+1$ cycles of the grid minor and thus the length of a shortest path between a pair of vertices, $x\in (S[i,j] \cap X_v)$ and $ y \in (S[i',j'] \cap X_v$), is at least $\Delta'+1$. This implies that the length of a shortest path between $x$ and $y$ in $P_G[X_v]$ is at least  $\Delta'+1$ 
and hence $|X_v|>\Delta'$, yielding the desired contradiction. By an identical argument one can show that, for any two internal pairs $(i,j)$ and $(i',j')$ such that $|i-i'|+|j-j'|>2\Delta'$ there is no edge with one endpoint in $S'[i,j]$ and the other in $S'[i',j']$. 

For every pair $a, b$ of non-negative integers such that $4\Delta' a + 2\Delta' \leq t$ and $4\Delta' b + 2\Delta' \leq t$ we define the sets 
\begin{itemize}\setlength\itemsep{-.7mm}
\item $V[a,b] = \bigcup_{i=0}^{2\Delta-1}\bigcup_{j=0}^{2\Delta-1} S[\Delta+4\Delta a+i,\Delta+4\Delta b+j]$.
\item $E_h[a,b] = \bigcup_{i=0}^{2\Delta-1} S[3\Delta+4\Delta a+i,2\Delta+4\Delta b]$.
\item $E_v[a,b] = \bigcup_{j=0}^{2\Delta-1} S[2\Delta+4\Delta a,3\Delta+4\Delta b+j]$.
\end{itemize}
One can think of each set $V[a,b]$ as a vertex of a grid, with each set $E_h[a,b]$ being a horizontal edge and each set $E_v[a,b]$ being a vertical edge in this grid. Build $V'[a,b]$ from $V[a,b]$ by replacing every vertex $v({\cal R}) \in V[a,b]$ by ${\cal V}({\cal R})$ and removing duplicates. Construct $E'_h[a,b]$ from $E_h[a,b]$ and $E'_v[a,b]$ from $E_v[a,b]$ similarly. We list the properties of the sets $V'[a,b]$, $E'_h[a,b]$ and $E'_v[a,b]$.
\begin{enumerate}\setlength\itemsep{-.7mm}
\item \label{pnt:connect} For every $a$, $b$, $G[V'[a,b]]$, $G[E'_h[a,b]]$ and $G[E'_v[a,b]]$ are connected. 
\item \label{pnt:disj1} Distinct sets $V'[a,b]$ and $V'[a',b']$ are pairwise disjoint, and there is no edge with one endpoint in $V'[a,b]$ and the other in $V'[a',b']$.
\item \label{pnt:disj2} For every $a$, $b$ the set $E'_h[a,b]$ is disjoint from every set $E'_h[a',b']$, $E'_v[a',b']$ and $V'[a',b']$, except possibly for $V'[a,b]$ and $V'[a+1,b]$.
\item \label{pnt:disj3} For every $a$, $b$ the set $E'_v[a,b]$ is disjoint from every set $E'_h[a',b']$, $E'_v[a',b']$ and $V'[a',b']$, except possibly for $V'[a,b]$ and $V'[a,b+1]$.
\item \label{pnt:adj} For every $a$, $b$ there is a vertex in $E'_h[a,b]$ which is adjacent to $V'[a,b]$ and a vertex which is adjacent to $V'[a+1,b]$. Furthermore there is a vertex in $E'_v[a,b]$ which is adjacent to $V'[a,b]$ and a vertex which is adjacent to $V'[a,b+1]$.
\end{enumerate}
Property~\ref{pnt:connect} follows directly from the fact that $P_G[V[a,b]]$, $P_G[E_h[a,b]]$ and $P_G[E_v[a,b]]$ are connected. Properties~\ref{pnt:disj1}, \ref{pnt:disj2} and~\ref{pnt:disj3} follow from the fact that for any two internal pairs $(i,j)$ and $(i',j')$ such that $|i-i'|+|j-j'|>2\Delta'$ the sets $S'[i,j]$ and $S'[i',j']$ are disjoint and have no edges between each other. Finally, Property~\ref{pnt:adj} follows from the fact that for every $a$,$b$ there is a vertex in $E_h[a,b]$ which is adjacent to $V[a,b]$ and a vertex which is adjacent to $V[a+1,b]$, and that there is a vertex in $E_v[a,b]$ which is adjacent to $V[a,b]$ and a vertex which is adjacent to $V[a,b+1]$.

For a pair $a$, $b$ of integers consider the set $E'_h[a,b]$. The properties~\ref{pnt:connect}, \ref{pnt:disj1} and~\ref{pnt:adj} ensure that some connected component $E^*_h[a,b]$ of $G[E'_h[a,b] \setminus (V'[a,b] \cup V'[a+1,b])]$ contains at least one neighbour of $V'[a,b]$ and one neighbour of  $V'[a+1,b]$. Similarly at least one connected component $E^*_v[a,b]$ of $G[E'_v[a,b] \setminus (V'[a,b] \cup V'[a,b+1])]$ contains at least one neighbour of $V'[a,b]$ and one neighbour of $V'[a,b+1]$. Then the family 
$$\{V'[a,b], E^*_h[a,b], E^*_v[a,b] : 4\Delta' a + 2\Delta' \leq t \mbox{ and } 4\Delta' b + 2\Delta' \leq t\}$$
of vertex sets in $G$ forms a model of a $\lfloor \frac{t-2\Delta'}{4\Delta'}\rfloor \times\lfloor \frac{t-2\Delta'}{4\Delta'}\rfloor$ grid minor in $G$ with every edge subdivided once. The sets $V'[a,b]$ are models of the vertices of the grid, the sets $E^*_h[a,b]$ are models of the subdivision vertices on the horizontal edges, while $E^*_v[a,b]$ are models of the subdivision vertices on the vertical edges.  Now by Lemma~\ref{lem:twofg}, we know that 
$\tw(P_G) \geq \frac{\tw(G) + 1}{(\Delta+1)} -1$. Combining this with the fact that $t=\frac{\tw(P_G)}{6}$ we can show that $G$ has a grid of size 
$\frac{\tw(G)}{100\Delta^3} \times \frac{\tw(G)}{100\Delta^3}$  as a minor. This concludes the proof.
\end{proof}

Using Lemma~\ref{lem:udglingrid} we show the following theorem and then combining this theorem with Observation~\ref{obs:timpliesd} we get an 
analogous result for $K_t$-free unit disk graphs.  

\begin{theorem} 
\label{thm:tstdiscgraphs} Let ${\cal G}_U^\Delta$ be the class of unit disk graphs such that the maximum degree of every graph $G$ in ${\cal G}_U^\Delta$ is at most $\Delta$. Then ${\cal G}_U^\Delta$ has truly sublinear treewidth with $\lambda = \frac{1}{2}$.
\end{theorem}
\begin{proof}
Let $G \in {\cal G}_U^\Delta$ have a vertex set $X$ such that $\tw(G \setminus X) \leq \eta$. Suppose for contradiction that $\tw(G) \geq (\eta+1)200\Delta^3\sqrt{|X|}$. Then, by Lemma ~\ref{lem:udglingrid}, $G$ contains a 
$(2(\eta+1)\sqrt{|X|}) \times (2(\eta+1)\sqrt{|X|})$ grid as a minor. This grid contains $4|X|$ vertex disjoint $(\eta+1) \times (\eta+1)$ grids. However, the set $X$ must intersect each of these grids, as $\tw(G \setminus X) \leq \eta$ and the treewidth of each grid is $\eta+1$. But then $|X| \geq 4|X|$, a contradiction.
\end{proof}

\begin{corollary} \label{cor:tstdiscgraphs} Let ${\cal G}_U^t$ be the class of unit disk graphs such that every graph $G$ in ${\cal G}_U^t$ is $K_t$ 
free. Then ${\cal G}_U^t$ has truly sublinear treewidth with $\lambda = \frac{1}{2}$.
\end{corollary}

\subsection{Structure of $K_t$-free Map Graphs.}
In this section we show that map graphs with bounded clique size have truly sublinear treewidth.
It is known that for every map graph $G$ 
one can associate a map ${\cal M} =(\mathscr{E} ,\omega )$ such that 
(i) no vertex in $\mathscr{E}$   is incident only to lakes;  
(ii)  there are no edges in $\mathscr{E}$ whose two incident faces are both lakes (possibly the same lake);
(iii) every vertex in $\mathscr{E}$ is incident to at most one lake, and incident to such a lake at most once ~\cite[p.~149]{DemaineHK09}. From now onwards we will assume that we are given map satisfying the above properties. For our proof we also need the following combinatorial lemma.


\begin{lemma}
\label{lem:mapboundeddegree}
Let $G$ be a map graph associated with $\cal M$ such that the maximum clique size in $G$ is at most $t$.
Then the maximum vertex degree of $\mathscr{E}$ is at most  $t+2$. 
\end{lemma}

\begin{proof}
Targeting towards a contradiction, let us assume that there is a vertex $v\in V(\mathscr{E})$ of degree at least 
$t+3$. By definition,  each connected component of $\mathscr{E}$ is 
biconnected and hence there are at least $t+2$ cyclic faces adjacent to $v$. By the properties of $\cal M$, we have that all, except maybe one,   adjacent faces are not  lakes. However the vertices corresponding to nation faces  form a clique of size $t+1$ in $G$, a contradiction.    \end{proof}

For our proof we also need the notions of radial   and dual of map graphs. 
The \emph{radial graph} $R = R({\cal M})$ has a vertex for every vertex of $\mathscr{E}$ and for every nation of $\mathscr{E}$, and $R$ 
is a bipartite graph with bipartition $V(\mathscr{E})$ and $N(\mathscr{E})$. Two vertices $v \in V(\mathscr{E})$ and 
$f\in N(\mathscr{E})$ are adjacent in $R$ if $v$ is incident to nation $f$.  
The  \emph{dual} $D = D({\cal M})$ of $\cal M$ has vertices corresponding only to the nations of $\mathscr{E}$.  The graph $D$ 
has a vertex for every nation of $\mathscr{E}$, and two vertices are adjacent in $D$ if 
the corresponding nations of $G$ share an edge. We now show a linear excluded grid theorem for map graphs with bounded maximum clique. 

\begin{lemma}
\label{lem:maplingrid}
There exists a constant $\rho$ such that any map graph $G$ with maximum clique size $t$ contains a 
$\frac{\rho \cdot \tw(G)}{t} \times \frac{\rho \cdot \tw(G)}{t}$ grid minor.
\end{lemma}
 \begin{proof}
 Let $\cal M$ be the map such that the graph associated with it is $G$. 
We now apply the result from~\cite[Lemma 4]{DemaineHK09} that states that the treewidth of the map graph $G$ is at most the product of the maximum vertex degree in $\mathscr{E}$ and $\tw(R)+ 1$. 
By Lemma~\ref{lem:mapboundeddegree}, we know that the maximum vertex degree of $\mathscr{E}$ is at most $t+2$, and hence 
$\tw(G)\leq (t+2)\cdot (\tw(R)+ 1)$. We now apply~\cite[Lemma 3]{DemaineHK09} which bounds the treewidth of a radial graph of a map. 
In particular, by~\cite[Lemma 3]{DemaineHK09} we have that for $R$, the radial graph of $\cal M$, $\tw(R)=O(\tw(D))$. This implies that 
$\tw(G)=O(t\cdot \tw(D))$. Observe that the graph $D$, the dual of $\cal M$, is a planar subgraph of $G$. By a result of Robertson 
et al.~\cite{RobertsonST94}, we have that for every planar graph $H$ there exists a constant $d$ such that it has 
$d\cdot \tw(H)\times d\cdot \tw(H)$ grid graph as a minor. This implies that there exists a constant $d$ such that $D$ has  
$d \cdot \tw(D)\times d \cdot \tw(D)$ grid graph as a minor. This combined with facts that $\tw(G)=O(t\cdot \tw(D))$ and $D$ is a subgraph of $G$ 
implies that there exists a constant $\rho$ such that $G$ has $\frac{\rho\cdot \tw(D)}{t}\times \frac{\rho \cdot \tw(G)}{t}$ grid graph as a minor. 
 \end{proof}

\begin{theorem}
\label{thm:tstmapgraphs}
Let ${\cal G}_M^t$ be the class of map graphs such that the every 
graph $G\in {\cal G}_M^t$ is $K_t$-free. Then ${\cal G}_M^t$ has truly sublinear treewidth with 
$\lambda = \frac{1}{2}$.
\end{theorem}

\begin{proof}
Let $G \in {\cal G}_M^t$ have a vertex set $X$ such that $\tw(G \setminus X) \leq \eta$. Suppose for contradiction that 
$\tw(G) \geq (\eta+1)2 \rho t \sqrt{|X|}$. Then, by Lemma ~\ref{lem:maplingrid}, $G$ contains a 
$(2(\eta+1)\sqrt{|X|}) \times (2(\eta+1)\sqrt{|X|})$ grid as a minor. This grid contains $4|X|$ vertex disjoint $(\eta+1) \times (\eta+1)$ grids. However, the set $X$ must intersect each of these grids, as $\tw(G \setminus X) \leq \eta$ and the treewidth of each grid is $\eta+1$. But then $|X| \geq 4|X|$, a contradiction.
\end{proof}


\subsection{$K_4$-Free Disc Graphs.}
Our result in Section~\ref{sec:udgoofbd} can easily be generalized to disk graphs of bounded degree. 
There is another widely used concept of {\emph{ply}} related to geometric graphs which has turned out 
be very useful algorithmically \cite{Leeuwen:2009ec}.
 An intersection graph $G$  generated by set of disks ${\cal B}=\{B_1,\ldots,B_n\}$  
(not necessarily unit disks) is said to have ply $\ell$ if every point in the plane is contained inside at most $\ell$ 
disks in ${\cal B}$. 
Observe that if a unit disk graph has bounded ply then it also has bounded vertex degree but this is not 
true for disk graphs. Here, we show that already disk graphs with ply $3$ and $K_4$-free disk graphs do not have truly sublinear treewidth. 

\begin{theorem}
\label{thm:k4freediscgraphs}
$K_4$-free disk graphs and disk graphs graphs with ply $3$ do not have truly sublinear treewidth. 
\end{theorem}
For the proof  of Theorem~\ref{thm:k4freediscgraphs},  we need the concept of a bramble.  
A {\em bramble} in a graph $G$ is a
family of connected subgraphs of $G$ such that any two of these subgraphs have a nonempty intersection or
are joined by an edge. The {\em order} of a bramble is the minimum number of vertices required to hit  all subgraphs
in the bramble. Seymour and Thomas~\cite{SeymourT89} proved that a graph has treewidth $k$ 
if and only if the maximum order of a bramble of $G$ is $k+1$. Thus a bramble of order 
$k+1$ is a witness that the graph has treewidth at least $k$. We will use this characterization to get a lower bound on the treewidth of the
graph we  construct. 

\begin{proof}
 We define a family $\cal F$ of disk graphs of ply $3$ such that for every $G\in \cal F$ we can find a set $X\subseteq V(G)$ such that 
$\tw(G\setminus X)\leq 1$ while $\tw(G)\geq |X|-1$. Given a natural number $t\geq 2$, our graph $G_t$ is defined as follows. We give the coordinates for centers of these disks. 
\begin{itemize}
\setlength{\itemsep}{-2pt}
\item We  have ``small"   disks of radius $0.99$ centered at $(1.25p,2q)$ for $0\leq p \leq 3t^2$ and $0 \leq q \leq t-1$. 
\item  We have ``large"  disks with radius $t-0.01$ centered at $((2p+1)t,t)$, $0\leq p\leq t-1$.
\end{itemize}
\begin{figure}
  \begin{center} 
    \includegraphics[scale=.2]{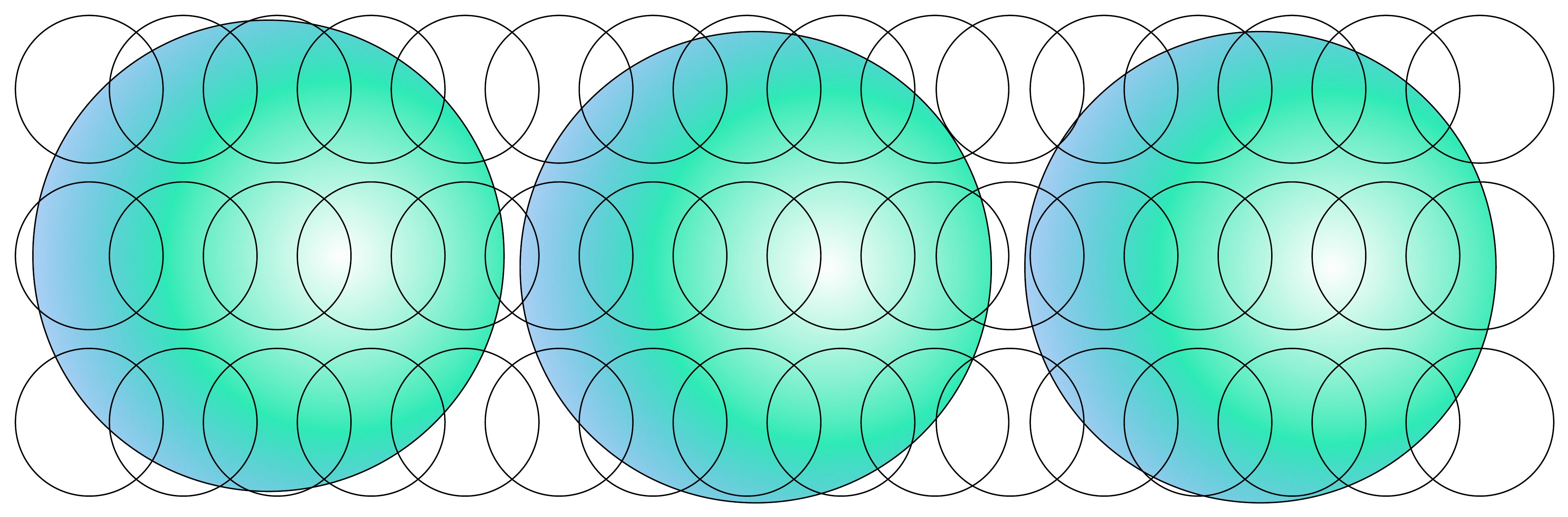}
   \end{center}
   \caption{{{Family of disks used to construct graph $G_t$ for $t=3$}}}
\label{ply3}
 \end{figure}
Intuitively, we have small disks stacked in $t$ rows, where in each row two consecutive disks intersect. Large
disks   intersect some unit disks in each row and they are pairwise disjoint among themselves. 
See Figure~\ref{ply3}, for an example of our construction. 
Let $G_t$ be 
the disk graph obtained from the intersection of disks placed as above. Observe that every point in the plane only occurs in at most 
$3$ disks and hence the ply of the graph is $3$.  Furthermore, since at most $3$ disks mutually intersect we have that $G$ is also $K_4$-free.  
Let $A$ be the set of vertices corresponding to small
disks in rows and $X$ be the set of remaining vertices. Observe that the graph induced by $A$ is a set of 
vertex disjoint paths and hence $\tw(G_t[A]) = \tw(G_t \setminus X) = 1$. 
We show that the treewidth of $G_t$ is  at least $t-1$ by exhibiting a  bramble of order $t$. Let us take the following set $S_i$, $0\leq i \leq t-1$. The set $S_i$ consists of vertices corresponding to small disks centered at $(1.25p,2i)$, where $0\leq p \leq 3t^2$ and a vertex corresponding to large disk with radius $t-0.01$ centered at $((2i+1)t,t)$. Since the disk with radius $t-0.01$ intersects at least one small disk in each row we have that the sets $S_i$ mutually intersect. Furthermore $S_i\cap S_j=\emptyset$ for all $i\neq j$. This implies that the smallest number of vertices required to cover all $S_i$ is at least $t$. This implies that $\tw(G)\geq t-1=|X|-1$.  
\end{proof}

\section{Applications}

In this section we show that every reducible minor-bidimensional problem with the separation property has EPTAS on 
unit disk graphs and map graphs with bounded maximum clique. Finally, we will show how we can 
extend these results to unit disk graphs and map graphs for variety of problems.  We recall that a PTAS is an algorithm which takes an instance $I$ of an optimization problem and a parameter $\epsilon>0$ and, runs in time $n^{\cO(f(1/\epsilon))}$, produces a solution that is within a factor $\epsilon$ of being optimal. A PTAS with running time $f(1/\epsilon)\cdot n^{\cO(1)}$, is called an efficient PTAS (EPTAS).

\paragraph{EPTAS on ${\cal G}_U^t$ and ${\cal G}_M^t$.}
Towards our goal we need the following result from~\cite{FLRSsoda2011}. 
\begin{proposition}[\cite{FLRSsoda2011} ]
\label{thm:eptastrans}
Let $\Pi$ be an $\eta$-transversable, reducible graph optimization problem. 
Then $\Pi$ has an EPTAS on every graph class ${\cal G}$ with truly  sublinear treewidth.
\end{proposition}

To use this result we only need to show that every reducible minor-bidimensional problem with the separation property is 
$\eta$-transversable for some $\eta$ on ${\cal G}_U^t$ and ${\cal G}_M^t$. For every fixed integer $\eta$ we define the {\sc $\eta$-Transversal} problem as follows. Input is a graph $G$, and the objective is to find a minimum cardinality vertex set $S \subseteq V(G)$ such that $\tw(G \setminus S) \leq \eta$. We now give a polynomial time constant factor approximation for the  {\sc $\eta$-Transversal} problem on ${\cal G}_U^t$ and ${\cal G}_M^t$. The 
proof of the following result is similar to the one in~\cite[Lemma 4.1]{FLRSsoda2011}.  We give the proof here for completeness. 
\begin{lemma}
\label{lem:transversapproxudcmap} 
For every integer $\eta$ there is a constant $c$ and a polynomial time $c$-approximation algorithm for the {\sc $\eta$-Transversal} problem on ${\cal G}_U^t$ and ${\cal G}_M^t$. 
\end{lemma}

For the proof of Lemma~\ref{lem:transversapproxudcmap} we also need the following result from ~\cite[Lemma 3.2]{FLRSsoda2011}. 

\begin{lemma}\label{lem:theLemma} Let ${\cal G}$ be a hereditary graph class of truly sublinear treewidth with parameter $\lambda$. For any $\epsilon < 1$ there is a $\gamma$ such that for any $G \in {\cal G}$ and $X \subseteq V(G)$ with $\tw(G \setminus X) \leq \eta$ there is a $X' \subseteq V(G)$ satisfying $|X'| \leq \epsilon|X|$ and for every connected component $C$ of $G \setminus X'$ we have $|C \cap X| \leq \gamma$ and
$|N(C)| \leq \gamma$.
Moreover $X'$ can be computed from $G$ and $X$ in polynomial time, where the polynomial is independent of  $\epsilon$, $\lambda$ and $\eta$. 
\end{lemma}

\begin{proof} [Proof of Lemma~\ref{lem:transversapproxudcmap}] 
Let $X$ be a smallest vertex set in $G$ such that $\tw(G \setminus X) \leq \eta$. By Lemma~\ref{lem:theLemma} with $\epsilon = 1/2$ there exists a $\gamma$ depending only on $t$ and $\eta$ and a set $X'$ with $|X'| \leq |X| / 2$ such that for any component $C$ of $G \setminus X'$, $|C \cap X| \leq \gamma$ and $|N(C)| \leq \gamma$. Since $X$ is the {\em smallest} set such that $\tw(G \setminus X) \leq \eta$, there is a component $C$ of $G \setminus X'$ with treewidth strictly more than $\eta$. Let $Z = N(C)$ and observe that $Z \subseteq X'$ is a set of size at most $\gamma$ such that $C$ is a connected component of $G \setminus Z$.

The algorithm proceeds as follows. It tries all possibilities for $Z$ and looks for a connected component $C$ of $G \setminus Z$ such that $\eta < \tw(G[C]) = O(\sqrt{\gamma})$. It solves the {\sc $\eta$-Transversal} problem optimally on $G[C]$ by noting that {\sc $\eta$-Transversal} can be formulated as a \pmin{} problem and applying the algorithm by Borie et al.~\cite{BoriePT92}. Let $X_C$ be the solution obtained for $G[C]$. The algorithm adds $X_C$ and $N(C)$ to the solution and repeats this step on $G \setminus (C \cup N(C))$ as long as $\tw(G) \geq \eta$. 

Clearly, the set returned by the algorithm is a feasible solution. We now argue that the algorithm is a $(\gamma+1)$-approximation algorithm. Let $C_1$, $C_2$, $\ldots ,C_t$ be the components found by the algorithm in this manner. Since $X$ must contain at least one vertex in each $C_i$ it follows that $t \leq |X|$. Now, for each $i$, $N(C_i)$ contains at most $\gamma$ vertices outside of $\bigcup_{j < i} N(C_j)$. Thus $\bigcup_{i \leq t} N(C_i) \leq \gamma|X|$. Furthermore for each $C$, $|X_C| \leq |X \cap C|$ and hence the size of the returned solution is at most $(\gamma+1)|X|$, which proves the lemma.
\end{proof}

By exchanging the Excluded grid theorem for $H$-minor free graphs by Linear Grid lemmas (Lemmas~\ref{lem:udglingrid} and ~\ref{lem:maplingrid}), 
we can adapt the proof of~\cite[Lemma 3.2]{F.V.Fomin:2010oq} to show the following. 
\begin{lemma}\label{lem:minor_bidem} Let $\Pi$ be a minor-bidimensional problem with the separation property. There exists a constant $\eta$ such that for every $G$ in ${\cal G}_U^t$ or ${\cal G}_M^t$, there is a subset $S \subseteq V(G)$  such that $|X|=O(t^{O(1)}\pi(G))$, and $\tw(G \setminus S)\leq \eta$. 
\end{lemma}
For the proof of the Lemma~\ref{lem:minor_bidem}, we also need the following well known lemma, see e.g. \cite{Bodlaender98},  on separators in graphs of bounded treewidth.  
\begin{lemma}
\label{lemma:balsep1}
Let $G$ be a graph of treewidth at most $t$ and $w : V(G) \rightarrow \{0,1\}$   
 be a weight function. Then there is a partition of $V(G)$ into $L \uplus S \uplus R$ such that
\begin{itemize}\setlength\itemsep{-1.2mm}
 \item $|S| \leq t+1$, $N(L) \subseteq S$ and $N(R) \subseteq S$,
 \item every connected component $G[C]$ of $G \setminus S$ has $w(C)\leq w(V)/2$,
 \item $\frac{w(V(G))-w(S)}{3}\leq w(L)\leq \frac{2(w(V(G))-w(S))}{3}$ and $\frac{w(V(G))-w(S)}{3}\leq w(R)\leq \frac{2(w(V(G)-w(S))}{3}$.
\end{itemize}
\end{lemma}

\begin{proof}[Proof of Lemma~\ref{lem:minor_bidem}]
We prove our result for graphs in ${\cal G}_U^t$. The result for graphs in ${\cal G}_M^t$ is analogous. Since $\Pi$ is a minor-bidimensional problem,  there exists a constant $\delta >0$ such that $\pi(R)\geq \delta r^2$ on $r\times r$ grid minor. The definition of 
minor-bidimensionality together with Lemma~\ref{lem:udglingrid} imply that for every $G \in {\cal G}_U^t$ we have that $\tw(G)\leq 10 t^3\sqrt{\frac{\pi(G)}{\delta}}$. Thus, there is a constant $d'$ depending on $t$ and $\delta$ such that $\tw(G)\leq d' \sqrt{\pi(G)}$. 

We first make the following observation.  For any numbers $a>0$, $b>0$, since $\lambda < 1$, we have  that $a^\lambda + b^\lambda > (a+b)^\lambda$. 
Thus we have $\rho = \min_{1/3 \leq \alpha \leq 2/3}\alpha^\lambda + (1-\alpha)^\lambda > 1$. Now we proceed to our proof. 

We want to construct a set $X$ such that the treewidth of $G[V(G)\setminus X]$ is at most $\eta$ (to be fixed later).  
Let us fix a solution $Z$ of size $\pi(G)$ for $G$ and a weight
function $w~:V(G)\rightarrow \{0,1\}$ that assigns $1$ to every vertex in $Z$ and $0$ otherwise. By Lemma~\ref{lemma:balsep1}, there is a partition of $V(G)$ into $L$, $S$ and $R$ such that $|S| \leq d'\sqrt{\pi(G)} +1 $, $N(L) \subseteq S$, $N(R) \subseteq S$, $|L \cap Z| \leq 2\pi(G)/3$ and 
$|R \cap Z| \leq 2\pi(G)/3$. By deleting $S$ from the graph $G$, we obtain   two graphs $G[L]$ and $G[R]$ with no edges between them. Since $\Pi$ is separable, there exists a constant $\beta$ such that $\pi(G[L]) \leq |Z\cap L|+\beta |S|$  and   $\pi(G[R]) \leq |Z\cap R|+\beta |S|$. Thus we put $S$ 
into $X$ and then proceed recursively in $G[L]$ and $G[R]$. Since $\Pi$ is minor-bidimensional problem with the separation property, we have 
that in recursive step for a graph $G'$ with solution of size $\ell$  we  find a separator of size $\cO(\sqrt \ell)$. We set $k=\pi(G)$. 
Then the size of the set $X$ we are looking for is governed by the following recurrence. 

\begin{align}
\label{equ1}
 &T(k) \leq \underset{\frac{1}{3}  \leq \alpha \leq \frac{2}{3}}{\max}\Big\{T(\alpha k + \beta \sqrt{k})+T((1-\alpha)k+ \beta \sqrt k)+ d' \sqrt{k}+1\Big\}.  
\end{align}
We first set two constants $q$ and $\gamma$ which will be used in the proof of above recurrence.  Set $q=\frac{6\beta+3d'}{\rho-1}$ and 
$\gamma=\lceil 4\delta^{-1} \rceil + (3\beta)^2 +(3q)^2+1$. The base case of the recursion is when $k\leq \gamma$ and once we reach this case, we 
do not decompose the graph any further. Thus for the base case we set $T(k)=0$. Now using induction we can show that the size of 
$|X|\leq k-q\sqrt{k}$ and for every component $C$ in $G[V(G)\setminus X]$ we have that $\pi(C)=O(\gamma)$.  We first show that if $k \geq \gamma / 3$ 
then $T (k) \leq k - q \sqrt{k}$ by induction on $k$. For the base case if $\gamma / 3 \leq k \leq \gamma$ then the choice of $\gamma$ implies that 
$k - q \sqrt{k} \geq  \frac{\gamma}{3} - q \sqrt{\gamma} \geq 0 = T(k)$.

We now consider $T(k)$ for $k > \gamma$. By our choice of $q$ we have that for all $k>\gamma$, $\alpha k+ \beta \sqrt{k} <k$ and 
$(1 - \alpha k) + \beta \sqrt{k} < k$. The induction hypothesis then yields the following inequality. 

\begin{align*}
T (k) & \leq & \max_{1/3 \leq \alpha \leq 2/3} & T( \alpha k + \beta \sqrt{k} + 1) + T((1-\alpha) k + \beta \sqrt{k} + 1) + d' \sqrt{k}+1  \\
& \leq & \max_{1/3 \leq \alpha \leq 2/3} & k - q \sqrt{\alpha k} - q (\sqrt{(1-\alpha)k)} + 2 (\beta \sqrt{k} + 1) + d' \sqrt{k}+1 \\
& \leq & \max_{1/3 \leq \alpha \leq 2/3} & k - q\sqrt{k}(\alpha^{1/2} + (1-\alpha)^{1/2}) + 2 (\beta \sqrt{k} + 1) + d' \sqrt{k}+1  \\
& \leq & &  k - q \sqrt{k} - q (\rho-1) \sqrt{k} + 2 (\beta \sqrt{k} + 1) + d' \sqrt{k}+1 \\
& \leq & &  k - q \sqrt{k} .
\end{align*}
The last inequality holds whenever $q  (\rho-1) \sqrt{k} \geq  2 (\beta \sqrt{k} + 1) + d' \sqrt{k}+1$, which is ensured by the choice of $q$ 
and the fact that $\sqrt{k}\geq 1$. Thus $T (k) \leq k$ for all $k$. This proves the bound on the size of $X$. Notice that every component 
$C$ in $G[V(G)\setminus X]$ has at most $\gamma $ vertices and thus by bidimensionality we have that $\tw(C)=O(\sqrt{\gamma})$. We set $\eta$ as  the $\mu \sqrt{\gamma}$, where $\mu$ is the constant appearing in the term $O(\sqrt{\gamma})$. This proves the theorem. 
\end{proof}

\begin{theorem}
\label{thm:udmapeptas} 
Let $\Pi$ be a reducible minor-bidimensional problem with the separation property. 
There is an EPTAS for $\Pi$ on ${\cal G}_U^t$ and ${\cal G}_M^t$.
\end{theorem}
\begin{proof}
Combining Lemmas \ref{lem:transversapproxudcmap} and  \ref{lem:minor_bidem} we get that every reducible minor-bidimensional problem $\Pi$ with the separation property is 
$\eta$-transversable for some $\eta$ on ${\cal G}_U^t$ and ${\cal G}_M^t$. Thus the theorem follows by applying 
Proposition~\ref{thm:eptastrans} in combination with Theorem~\ref{thm:tstmapgraphs} and Corollary~\ref{cor:tstdiscgraphs}. 
\end{proof}

\paragraph{EPTAS on Unit Disc Graphs and Map Graphs.}
In this section we give EPTAS for several problems on unit disk graphs and map graphs. Our first problem is the following generic problem. Let  $\cal F$ be a finite set of graphs. In the {\fd} problem, we are given an $n$-vertex graph
$G$ as an input, and asked to find a minimum sized subset $S\subseteq V(G)$ 
such that $G\setminus S$ does not contain a graph from  ${\cal
F}$ as a minor. We refer to such a subset $S$ as a $\cal F$-hitting set.   
The {\fd}~problem is a generalization of several fundamental problems. For example, when  ${\cal 
F}=\{K_2\}$, a complete graph on two vertices, this is the 
{\sc Vertex Cover} problem. When ${\cal
F}=\{C_3\}$,   a cycle on three vertices, this is the 
   {\sc Feedback Vertex Set} problem.  Other famous cases are  ${\cal
F} =\{K_{2,3}, K_4\}$, ${\cal
F} =\{K_{3,3}, K_5\}$ and ${\cal
F} =\{K_{3}, T_2\}$, which correspond to removing vertices to obtain outerplanar graphs, planar 
graphs, and graphs of pathwidth one respectively.  Here, $K_{i,j}$ denotes the complete bipartite graph with bipartitions of sizes $i$ and $j$, and $K_i$ denotes the complete graph on $i$ vertices. Further, a $T_2$ is a star on three leaves, each of whose edges has been subdivided exactly once. 
In literature, these problems are known as   {\sc Outerplanar Deletion Set}, {\sc Planar Deletion Set} and 
{\sc Pathwidth One Deletion Set} respectively. Now we show that if $\cal F$ contains a planar graph then {\fd} has EPTAS on unit disk graphs and map graphs. It is known from~\cite{FLRSsoda2011} that  {\fd} problem is reducible minor-bidimensional problem with the separation property whenever $\cal F$ contains  a planar graph. 

\begin{theorem}
\label{lem:fdeptas}
 Let $\cal F$ be  a finite set of graphs containing a planar graph. 
   Then  {\fd}   admits  an EPTAS on unit disk graphs and map graphs. 
\end{theorem}
\begin{proof}
Let  $G$ be the input graph, $\epsilon $ be a fixed constant and 
$\cal F$ be an obstruction set containing a planar graph of size $h$. 
This implies that any optimal $\cal F$-hitting set in $G$ must contain all but $h-1$ 
vertices from any clique in $G$.  We outline a proof below only for unit disk graphs, the proof for map graphs is similar. 

The algorithm proceeds as follows. 
It finds a maximum clique 
$C$ of $G$. One can find a maximum sized clique in unit disk graphs and map graphs in polynomial time~\cite{ChenGP98,ChenGP02,ClarkCJ90,RaghavanS03}. 
The algorithm adds $C$ to the solution and repeats this step on $G \setminus C$ as long 
as there is a  clique of size $ \frac{(1 + \epsilon)h}{\epsilon}$. Once we have that the maximum size of a 
clique is bounded by $\frac{(1 + \epsilon)h}{\epsilon}$,  
we can use the EPTAS obtained in Theorem~\ref{thm:udmapeptas} to get  a $\cal F$-hitting set of $G$ of size  
$(1+\epsilon)OPT$, where $OPT$ is the size of a minimum $\cal F$-hitting set.

Clearly, the set returned by the algorithm is a feasible solution. We now argue that the algorithm is an EPTAS.
Clearly the running time of our algorithm is of desired form. As the step where we find a clique and add all its vertices in our solution can be done in polynomial time 
and finally we run an EPTAS on a graph where the maximum degree is bounded by a function of $\epsilon$.   
Let $X$ be an optimal $\cal F$-hitting set of $G$. 
Let $C_1$, $C_2$, $\ldots ,C_t$ be the cliques found by the algorithm and $G_q$ be the graph where we apply Theorem~\ref{thm:udmapeptas}. 
Since $X$ must contain at least $|C_i|-h$ vertices  and  $|C_i|\geq \frac{(1 + \epsilon)h}{\epsilon}$, we have that $|C_i|\leq (1+\epsilon)(|C_i|-h)\leq (1+\epsilon)(|X\cap C_i|)$. 
Thus the size of the solution returned by the algorithm satisfies the following inequality $\sum_{i=1}^t |C_i|+(1+\epsilon) |X\cap V(G_q)| 
\leq (1+\epsilon)(\sum_{i=1}^t |X\cap C_i| +  |X\cap V(G_q)|) \leq (1+\epsilon)|X|=(1+\epsilon)OPT$. This completes the proof. 
\end{proof}
Next we show how we can obtain an EPTAS for {\sc Connected Vertex Cover} on unit disk graphs and map graphs. In {\sc Connected Vertex Cover} we are given a graph $G$ and the objective is to find a minimum size subset $S \subseteq V(G)$ such that $G[S]$ is connected and every edge in $E(G)$ has at least one endpoint in $S$.  

\begin{theorem}
\label{lem:cvcunitmap}
  {\sc Connected Vertex Cover}  admits  an EPTAS on unit disk graphs and map graphs. 
  \end{theorem}
  \begin{proof}
Observe that {\sc Connected Vertex Cover} is $0$-transversable. Given a graph $G$ we find a maximal matching in linear time and output the endpoints of the matching as $X$. Any vertex cover must contain at least one endpoint from each edge in the matching, and thus $|X| \leq 2\pi(G)$. Also, $\tw(G \setminus X) = 0$. We recall the proof from~\cite{FLRSsoda2011}  that  {\sc Connected Vertex Cover} is reducible, as we will use this to get 
an EPTAS here.   
Given a graph $G$ and set $X$, let $G' = G \setminus X$ and let $R = N(X)$.  The annotated problem $\Pi'$ is to find a minimum sized set $S' \subseteq V(G')$ such that every edge in $G'$ has an end point in $S'$ and every connected component of $G'[S']$ contains a vertex in $R$. Notice that for any connected vertex set $S$ of $G$, $S \setminus X$ is a feasible solution to $\Pi'$ on $G'$. Conversely, for any feasible solution $S'$ of $\Pi'$ on $G'$, we have that $S = S' \cup X$ is a vertex cover of $G$ and has at most 
$|X|$ connected components. Since $S$ is a vertex cover it is sufficient to add $(|X|-1)$ vertices to $S$ in order to make it a connected vertex cover of $G$. Hence, {\sc Connected Vertex Cover} is reducible. One can similarly show that the annotated problem 
$\Pi'$ is $0$-transversable and reducible. This implies that $\Pi'$ has EPTAS on ${\cal G}_U^t$ and ${\cal G}_M^t$. 

To get our EPTAS for {\sc Connected Vertex Cover} we do similar to what we did for {\fd} problem in Theorem~\ref{lem:fdeptas}. The only change 
is that we keep finding clique and including it in our solution until there is no clique of size 
$ \frac{(2 + \epsilon)}{\epsilon}$. 
Let $C_1$, $C_2$, $\ldots ,C_q$ be the cliques found by the algorithm and $G_q$ be the graph on which we apply Proposition~\ref{thm:eptastrans}. Let $Z$ be the union of cliques, that is, $Z=\cup_{i\leq q}C_i$. Now we define the annotated problem $\Pi'$ with respect to set $Z$ and using Proposition~\ref{thm:eptastrans} obtain a set $W$ of size $(1+\epsilon)OPT'$, where $OPT'$ is the size of a minimum cardinality set in $G_q$ such that every edge in $G_q$ has 
an end point in $W$ and every connected component of $G_q[W]$ contains a vertex in $R=N(Z)\cap V(G_q)$. Now consider the set $W\cup Z$. This is a vertex cover of $G$ such that it has $q$ components and hence we can make it connected by adding at most 
$q-1$ vertices. Let the final solution returned by our algorithm be $S$. Let $X$ be an optimal connected vertex cover of $G$. 
Since $X$ must contain at least $|C_i|-1$ vertices  and 
the size of $|C_i|\geq \frac{(2 + \epsilon)}{\epsilon}$, we have that $|C_i|+1\leq (1+\epsilon)(|C_i|-1)\leq (1+\epsilon)(|X\cap C_i|)$. 
Thus the size of the solution returned by the algorithm satisfies the following inequality $\sum_{i=1}^t (|C_i|+1)+(1+\epsilon) 
|X\cap V(G_q)| \leq (1+\epsilon)(\sum_{i=1}^t (|X\cap C_i|) +  |X\cap V(G_q)|) \leq (1+\epsilon)|X|=(1+\epsilon)OPT$. This completes the proof. 
\end{proof}

\paragraph{EPTAS for {\sc Cycle Packing} on Unit Disk Graphs.}
A {\em cycle packing} in a graph $G$ is a collection $C_1, C_2, \ldots , C_p$ of pairwise disjoint vertex sets such that for every $i$, $G[C_i]$ induces a cycle. The integer $p$ is the {\em size} of the cycle packing and in the {\sc Cycle Packing} problem the objective is to find a cycle packing of maximum size in the input graph. {\sc Cycle Packing} is known to be minor-bidimensional, separable and reducible~\cite{FLRSsoda2011}. Thus  by Theorem~\ref{thm:udmapeptas}, the problem admits an EPTAS on ${\cal G}_U^t$. Hence, in order to give an EPTAS for {\sc Cycle Packing} on unit disk graphs it is sufficient to prove the following lemma. In particular, the following lemma implies that if we find a sufficiently large clique $X$, partition $X$ into triangles and add this partition to our packing, this will give a good approximation of how the optimum solution intersects with $X$. Here a {\em triangle} is a cycle on three vertices.

\begin{lemma}
\label{lem:cyclepackclique} Let $G$ be a unit disk graph and $X$ be a clique in $G$. There is a maximum size cycle packing $C_1, C_2, \ldots ,C_p$ in $G$ such that at most $405$ cycles $C_i$ in the packing satisfy $C_i \cap X \neq \emptyset$ and $C_i \setminus X \neq \emptyset$. \end{lemma}
\begin{proof}
Let $X$ be a clique in $G$. The centers of all disks corresponding to vertices of $X$ must be inside a $2 \times 2$ square. Thus the centers of all disks corresponding to vertices in $N(X)$ must be in a $6\times 6$ square. By~\cite[Lemma 2]{DumPach09}, the vertices in $N(X)$ can be partitioned into $27$ cliques $S_1$, $S_2$, \ldots $,S_{27}$. Note that in the definition of unit disk graphs used in ~\cite[Lemma 2]{DumPach09} two vertices are adjacent if the centers of the corresponding disks is at distance at most $1$ from each other, while in this paper two vertices are adjacent if the centers of their disks are at distance at most $2$. This difference is taken into account when applying \cite[Lemma 2]{DumPach09}. We say that a cycle $C$ {\em crosses} $X$ if $C \cap X \neq \emptyset$ and  $C \setminus X \neq \emptyset$. Let $C_1, C_2, \ldots, C_p$ be a maximum cycle packing in $G$ that has the fewest cycles 
crossing $X$. Observe that any cycle $C$ that crosses $X$ intersects with $X$ in at most two vertices --- since otherwise $G[C \cap X]$ induces a triangle, say $T$  and then we can replace $C$ by $T$ in the cycle packing and  obtain a maximum size cycle packing with fewer cycles that cross $X$.  This contradicts the choice of the packing $C_1, C_2, \ldots, C_p$.

We prove that there can be at most $54$ cycles in the packing that intersect $X$ in exactly $2$ vertices. Suppose for contradiction that there are at least $55$ such cycles. Each such cycle contains at least one vertex in $N(X)$. Since each vertex in $N(X)$ is in one of the $27$ cliques $S_1, \ldots ,S_{27}$ the pigeon hole principle implies that there are three cycles $C_a$, $C_b$ and $C_c$ in the packing which all intersect $X$ in exactly two vertices and a clique $S_i$ such that $C_a \cap S_i \neq \emptyset$, $C_b \cap S_i \neq \emptyset$ and $C_c \cap S_i \neq \emptyset$. Since all cycles in the packing are vertex disjoint, this means that $S_i \cap (C_a \cup C_b \cup C_c)$ contains a triangle $T_1$. On the other hand, $X \cap (C_a \cup C_b \cup C_c)$ is a clique on $6$ vertices, and can be partitioned into two triangles $T_2$ and $T_3$. Now we can remove $C_a$, $C_b$ and $C_c$ from the proposed packing and replace them by $T_1$, $T_2$ and $T_3$. The resulting packing has the same size, but fewer cycles that cross $X$. This contradicts the choice of the packing $C_1, C_2, \ldots, C_p$.

Now we show that there can be at most $2(27\times 27)=  1458$ cycles in the packing that intersect with $X$ in exactly $1$ vertex. Every such cycle contains at least two vertices in $N(X)$. For a pair $(i,j)$ of integers $1 \leq i \leq j \leq 27$ we say that a cycle $C_a$ is an $(i,j)$ cycle if $C_a$ contains two distinct vertices $u$ and $v$ such that $u \in S_i$ and $v \in S_j$. If there are more than $1458$ cycles in the packing that intersect with $X$ in exactly $1$ vertex then there are $i$ and $j$ such that there are three $(i,j)$-cycles $C_a$, $C_b$ and $C_c$ in the packing that intersect $X$ in one vertex. Let $u_a$, $u_b$ and $u_c$ be three vertices in $C_a \cap S_i$,  $C_b \cap S_i$ and  $C_c \cap S_i$ respectively. Similarly, let $v_a$, $v_b$ and $v_c$ be the three vertices in $C_a \cap S_j$,  $C_b \cap S_j$ and  $C_c \cap S_j$ respectively. Now $T_1 = \{u_a,u_b,u_c\}$, $T_2 = \{v_a,v_b,v_c\}$ and $T_3 = X \cap (C_a \cup C_b \cap C_c)$ are vertex disjoint triangles. We can replace $C_a$, $C_b$, and $C_c$ by $T_a$, $T_b$ and $T_c$ in the cycle packing and obtain a maximum size cycle packing with fewer cycles that cross $X$, contradicting the choice of $C_1, \ldots ,C_p$. Hence there are at most $27+1458=1485$ cycles in the packing that cross $X$.
\end{proof}

\begin{theorem} 
\label{lem:CP}
{\sc Cycle Packing} admits an EPTAS on unit disk graphs.\end{theorem}
\begin{proof}
Given a unit disk graph $G$ and $\epsilon$, we choose $t$ to be $\frac{(1485\times 3)=4455}{\epsilon}$. If $G$ does not contain a clique of size $t$ then we apply the EPTAS for {\sc Cycle Packing} on ${\cal G}_U^t$ guaranteed by Theorem~\ref{thm:udmapeptas} to give a $(1-\epsilon)$-approximation for {\sc Cycle Packing}. If $G$ contains a clique $X$ of size $t$, the algorithm partitions $X$ into $\frac{|X|}{3}$ triangles $T_1, \ldots T_x$, recursively finds a $(1-\epsilon)$-approximate cycle packing $C_1, \ldots C_p$ in $G \setminus X$ and returns  $T_1,\ldots T_x, C_1, \ldots C_p$ as an approximate solution. Clearly, the algorithm terminates in $f(\epsilon)\cdot n^{O(1)}$ time, so it remains to argue that the returned solution is indeed a  $(1-\epsilon)$-approximate cycle packing of $G$. We prove this by induction on the number $n$ of vertices in $G$. Let $OPT$ be the size of the largest cycle packing in $G$.

If there is no clique of size $t$  and we apply the EPTAS for {\sc Cycle Packing} on ${\cal G}_U^t$ then clearly the returned solution is a $(1-\epsilon)$-approximation. If the algorithm finds such a clique $X$, Lemma~\ref{lem:cyclepackclique} ensures that there is a cycle packing of size $OPT$ such that at most $1485$ cycles in the packing cross $X$. All cycles in the packing that intersect with $X$ but do not cross $X$ are triangles in $X$. Hence $G \setminus X$ contains a cycle packing of size at least $OPT-\frac{|X|}{3}-1485$. By the inductive hypothesis the algorithm returns a cycle packing in $G \setminus X$ of size at least $(OPT-\frac{|X|}{3}-1485)(1-\epsilon)$. Now, $X$ contains $\frac{|X|}{3}$ triangles $T_1, \ldots T_x$. Hence, the total size of the packing returned by the algorithm is at least
$$\Big(OPT-\frac{|X|}{3}-1485\Big)(1-\epsilon) + \frac{|X|}{3} = OPT(1-\epsilon) -\Big (\frac{|X|}{3}+1485\Big)(1-\epsilon) + \frac{|X|}{3} \geq OPT(1-\epsilon)$$
since $|X| \geq t$. This concludes the proof.
\end{proof}

\paragraph{EPTAS for  {\sc (Connected) Vertex Cover} 
on Unit Ball Graphs in $\mathbb{R}^d$.}
Our results in this section are based on an observation that if for some graph class $\cal G$ the size of an optimum solution for a problem $\Pi$ 
and the number of vertices in the input graph are linearly related then to obtain EPTAS it is sufficient that $\cal G$ has sublinear 
treewidth rather than truly sublinear treewidth. The crux of this result is based on the following adaptation of the decomposition 
lemma proved in~\cite[Lemma 3.2]{FLRSsoda2011}. 
\begin{lemma}
\label{lem:thenewLemma} Let ${\cal G}$ be a hereditary graph class of sublinear treewidth with parameter $\lambda<1$, that is, 
for every $G\in \cal G$, $\tw(G)=O(|V(G)|^\lambda)$.  
For every $\epsilon < 1$ there is   $\gamma$ such that for any $G \in {\cal G}$ 
there is   $X \subseteq V(G)$ satisfying $|X| \leq \epsilon |V(G)|$ and for every connected component $C$ of $G \setminus X$,  
we have that $|C|\leq \gamma$. 
Moreover $X$ can be computed from $G$  in polynomial time.
\end{lemma}

\begin{proof} Let the number of vertices of $G$ be $n$, that is, $|V(G)|=n$. 
For any $\gamma \geq 1$, define $T_\gamma : \mathbb{N} \rightarrow \mathbb{N}$ such that $T_\gamma(n)$ is the smallest integer such that if 
$G \in {\cal G}$ and $|V(G)| \leq n$, then there is a $X \subseteq V(G)$ of size at most $T_\gamma(n)$ such that for 
every connected component $C$ of $G \setminus X$ we have $|C | \leq \gamma$. 
Informally, $T_\gamma(n)$ is the minimum size of a vertex set $X$ such that every connected component $C$ of $G \setminus X$ has at most 
$\gamma$ vertices. Furthermore, since ${\cal G}$ is a hereditary graph class of sublinear treewidth with parameter $\lambda$ there exists a constant $\beta$ such that $\tw(G) \leq  \beta n^\lambda$.  We will make choices for the constants $\delta$ and $\gamma$ and $\rho$ based on 
 $\lambda$, $\beta$ and $\epsilon$.  Our aim is to show that $T_\gamma(n) \leq \epsilon n$ for every $n$.

Observe that for any numbers $a>0$, $b>0$, we have $a^\lambda + b^\lambda > (a+b)^\lambda$ since $\lambda < 1$. Thus we have $\rho = \min_{1/3 \leq \alpha \leq 2/3}\alpha^\lambda + (1-\alpha)^\lambda > 1$. We choose $\delta = \frac{(2\epsilon+1)(\beta+1)}{\rho - 1}$ and $\gamma=(\frac{3\delta}{\epsilon})^{\frac{1}{1-\lambda}}$. 
If $|V(G)| \leq \gamma$ then we set $X=\emptyset$, so $T_\gamma(n) = 0 \leq \epsilon n$ for $n \leq \gamma$.  We now show that if $n \geq \gamma / 3$ then $T_\gamma(n) = 0 \leq \epsilon n- \delta n^\lambda$ by induction on $n$. For the base case if $\gamma / 3 \leq n \leq \gamma$ then the choice of $\gamma$ implies the following inequality.
$$\epsilon n - \delta n^\lambda \geq \epsilon \frac{\gamma}{3} - \delta \gamma^\lambda \geq 0 = T_\gamma(n)$$

We now consider $T_\gamma(n)$ for $n > \gamma$.  We know that the treewidth of $G$ is at most $ \beta n^\lambda$. Construct a weight function $w : V(G) \rightarrow \mathbb{N}$ such that $w(v)=1$,  for all $v\in V(G)$. By Lemma~\ref{lemma:balsep1}, there is a partition of $V(G)$ into $L$, $S$ and $R$ such that $|S| \leq  \beta n^\lambda + 1$, $N(L) \subseteq S$, $N(R) \subseteq S$, $|L | \leq 2n/3$ and $|R | \leq 2n/3$. Deleting $S$ from the graph $G$ yields two graphs $G[L]$ and $G[R]$ with no edges between them. Thus we put $S$ into $X$ and then proceed recursively in 
$G[L ]$ and $G[R ]$. 
This yields the following recurrence for $T_\gamma$.
$$
T_\gamma(n) \leq \max_{1/3 \leq \alpha \leq 2/3} T(\alpha n +  \beta n^\lambda + 1) + T((1-\alpha) n +  \beta n^\lambda + 1) + \beta n^\lambda + 1.
$$
Observe that since $n \geq \gamma$ we have $\alpha n \geq \gamma / 3$ and $(1 - \alpha n) \geq \gamma / 3$. The induction hypothesis then yields the following inequality.
%
\begin{align*}
T_\gamma(n) & \leq & \max_{1/3 \leq \alpha \leq 2/3} & T(\alpha n +  \beta n^\lambda + 1) + T((1-\alpha) n +  \beta n^\lambda + 1) + 
  \beta n^\lambda + 1 \\
& \leq & \max_{1/3 \leq \alpha \leq 2/3} & \epsilon n - \delta (\alpha n)^\lambda - \delta((1-\alpha)n)^\lambda + (2\epsilon+1)(\beta n^\lambda + 1) \\
& \leq & \max_{1/3 \leq \alpha \leq 2/3} & \epsilon n - \delta n^\lambda(\alpha^\lambda + (1-\alpha)^\lambda) + (2\epsilon+1)(\beta n^\lambda + 1) \\
& \leq & & \epsilon n - \delta n^\lambda - \delta (\rho-1) n^\lambda + (2\epsilon +1)(\beta n^\lambda +  1) \\
& \leq & & \epsilon n - \delta n^\lambda.
\end{align*}
%
The last inequality holds whenever $\delta (\rho-1) n^\lambda \geq (2\epsilon +1)(\beta n^\lambda +  1)$, which is ensured by the choice of $\delta$ and the fact that $n^\lambda \geq 1$. Thus $T_\gamma(n) \leq \epsilon n$ for all $n$. Hence there exists a set $X$ of size at most $\epsilon n$ such that for every component $C$ of $G \setminus X$ we have $|C|\leq \gamma$.


What remains is to show that $X$ can be computed from $G$ in polynomial time. The inductive proof can be converted into a recursive algorithm. The only computationally hard step of the proof is the construction of a tree-decompositon of $G$ in each inductive step. Instead of computing the treewidth exactly we use the $d^*\sqrt{\log \tw(G)}$-approximation algorithm by Feige et al.~\cite{FeigeHajLee05}, where $d^*$ is a fixed constant. Thus when we partition $V(G)$ into $L$, $S$, and $R$ using Lemma~\ref{lemma:balsep1}, the upper bound on the size of $S$ will be 
$d^*( \beta n^\lambda+1)\sqrt{\log(\beta n^\lambda)}$ instead of $ \beta n^\lambda+1$. However, for any $\lambda < \lambda' < 1$ there is a $\beta'$ such that  $d^*( \beta n^\lambda)\sqrt{\log( \beta n^\lambda)} < \beta' n^{\lambda'}$. Now we can apply the above analysis with $\beta'$ instead of $\beta$ and $\lambda'$ instead of $\lambda$ to bound the size of the set $X$ output by the algorithm. This concludes the proof of the lemma.
\end{proof}

\noindent
Using Lemma~\ref{lem:thenewLemma} we can obtain the following analogue of Proposition~\ref{thm:eptastrans} (\cite[Theorem 4.1]{FLRSsoda2011}).
\begin{theorem}
\label{thm:eptastransnew}
Let $\Pi$ be  a  
reducible graph optimization problem 
and let ${\cal G}$ be a class of graphs with sublinear treewidth such that  for every $G\in \cal G$, $\pi(G)=\Omega(|V(G)|)$.
Then $\Pi$ has an EPTAS on  ${\cal G}$.
\end{theorem}
\begin{proof}
Let $G$ be the input to $\Pi$, $|V(G)|=n$ and $\epsilon > 0$ be fixed. Since for every $G\in \cal G$, $\pi(G)=\Omega(n)$, we have 
that $\pi(G)\geq \rho_1 n$, for a fixed constant $\rho_1$.    
Furthermore, since ${\cal G}$ is a hereditary graph class of sublinear treewidth with parameter $\lambda$, there exists a constant $\beta$  
such that  $\tw(G) \leq  \beta n^\lambda$.  Let $\epsilon'$ be a constant to be selected later. By Lemma~\ref{lem:thenewLemma}, there exist $\gamma$, $\lambda' < 1$ and $\beta'$ depending on $\epsilon'$, $\lambda$ and $\beta$ such that given $G$  a set $X$ with the following properties can be found in polynomial time. First $|X| \leq \epsilon' n$, and secondly for every component $C$ of $G \setminus X$ we have that 
$|C|  \leq \gamma$. Thus $\tw(G \setminus X) = \tau \leq \beta' \gamma^{\lambda'} $.  
Since $\Pi$ is reducible, there exists  a \pmm{} problem $\Pi'$, a constant $\rho_2$ and a function $f : \mathbb{N} \rightarrow \mathbb{N}$ such that: 
\begin{enumerate}\setlength\itemsep{-1.2mm}
 \item there is a polynomial time algorithm that given $G$ and $X \subseteq V(G)$ outputs $G'$ such that $|\pi'(G') - \pi(G)| \leq \rho_2 |X|$ and $\tw(G') \leq f(\tau)$,
 \item there is a polynomial time algorithm that given $G$ and $X \subseteq V(G)$, $G'$ and a vertex (edge) set $S'$ such that $P_{\Pi'}(G',S')$ holds outputs $S$ such that $\phi_\Pi(G,S)$ holds and $|\kappa_\Pi(G,S)-|S'|| \leq \rho_2 |X|$.
\end{enumerate}
We constuct $G'$ from $G$ and $X$ using the first polynomial time algorithm. Since $\tw(G') \leq f(\tau)$ we can use an extended version of Courcelle's theorem~\cite{Courcelle90,Courcelle97}  given by  Borie et al.~\cite{BoriePT92} to find an optimal solution $S'$ to $\Pi'$ in $g(\epsilon')|V(G')|$ time. By the properties of the first polynomial time algorithm, $||S'|-\pi(G)| \leq \rho |X|$ where $\rho = \max(\rho_1,\rho_2)$. We now use the second polynomial time algorithm to construct a  solution $S$ to $\Pi$ from $G$, $X$, $G'$ and $S'$. The properties of the second algorithm ensure $\phi_\Pi(G,S)$ holds and that $|\kappa_\Pi(G,S)-|S'|| \leq \rho |X|$, and hence $|\kappa_\Pi(G,S)-\pi(G)| \leq 2 \rho |X| \leq 2 \rho^2 \epsilon' \pi(G)$. Choosing $\epsilon' = \frac{\epsilon}{2\rho^2}$ yields $|\kappa_\Pi(G,S)-\pi(G)| \leq \epsilon \pi(G)$, proving the theorem.
\end{proof}


\begin{lemma}
\label{lem:boundingsolutionsize}
 Let $G$ be an intersection graph of unit balls in $\mathbb{R}^d$, for a fixed $d$. 
 If $G$ does not contain an isolated vertex then the size of minimum (connected) 
vertex cover is at least $|V(G)|/f(d)$, where $f(d)=2(2^{0.401d(1+o(1))}+1)$. 
\end{lemma}
\begin{proof}
Let $G$ be an intersection graph of unit balls in $\mathbb{R}^d$, for a fixed $d$. For our proof we need the 
concept of {\em kissing number}. The kissing number $\tau_d$ is the maximum number of 
non overlapping $d$-dimensional unit balls of equal size that can touch a unit ball in $\mathbb{R}^d$.  
Kabatiansky and Levenshtein~\cite{KabatianskyL78} showed that $\tau_d \leq 2^{0.401d(1+o(1))}$. This implies that for any 
vertex $v\in V(G)$, $N(v)$ does not contain an independent set of size bigger than $\tau_d+1$.  
 
Given a graph $G$ we compute a maximal matching, say $M$. 
Clearly the size of $M$ is a lower bound on the size of a 
minimum (connected) vertex cover. Let $V(M)$ be the set of end points of edges in $M$ and $I=V(G)\setminus V(M)$. 
Clearly $I$ is an independent set. Furthermore every vertex in $I$ is adjacent to some vertex in $V(M)$. Hence 
we have that $|I|\leq |V(M)|(\tau_d+1)$. This implies that $|V(G)|=|V(M)|+|I|\leq 2|M|+2|M|(\tau_d+1)$. The last inequality 
implies the lemma.
\end{proof}

Finally, we note that every graph $G$, that is, an intersection graph of unit balls in $\mathbb{R}^d$, with 
maximum clique size $\Delta$ has the property that every point in $\mathbb{R}^d$ is in at most $\Delta$ 
unit balls. This together with result from~\cite{MillerTTV97} implies that the treewidth of $G$ is $c_d\Delta^{1/d}|V(G)|^{1-\frac{1}{d}}$, where $c_d$ is a constant depending only on $d$.  
This implies that  an intersection graph of unit balls in $\mathbb{R}^d$ with bounded maximum clique has sublinear treewidth. So an 
EPTAS for {\sc Connected Vertex Cover} and {\sc Vertex Cover} can be obtained along the similar lines   
as in Theorems~\ref{lem:fdeptas} and \ref{lem:cvcunitmap} and finally uses Theorem~\ref{thm:eptastransnew} instead  of 
Theorem~\ref{thm:eptastrans} to arrive to the following result. 
\begin{theorem}
\label{lem:alldvccvc}
 {\sc Connected Vertex Cover} and {\sc Vertex Cover} admit an EPTAS on unit ball graphs of fixed dimension. 
\end{theorem}
We can also obtain EPTASs for {\sc Connected Vertex Cover} and {\sc Vertex Cover} on disk graphs, as on disk 
graphs of bounded clique size we have that the size of an optimum solution 
and the number of vertices in the input graph are linearly related. 

\paragraph{Parameterized Subexponential Time Algorithms.}
In this section we show how to obtain parameterized subexponential time algorithm for several problems. 
Formally, a {\em parameterization}
of a problem is assigning an integer $k$ to each input instance and
a parameterized problem is {\em fixed-parameter tractable} {\sf  (FPT)} if there is an algorithm that solves the problem in time
$f(k)\cdot |I|^{O(1)}$, where $|I|$ is the size of the input and $f$ is an
arbitrary computable function.  
We say that a parameterized problem has a parameterized subexponential algorithm if it is solvable in time $2^{o(k)}\cdot |I|^{O(1)}$

Our basic 
idea is to find a ``large'' clique and guess the intersection of an optimal solution with this clique. 
We recursively do this until we do not have 
a large clique. Once we do not have large clique we have that the maximum degree of the graph is bounded and thus the 
treewidth comes into picture. At that point we use dynamic programming on graphs of bounded treewidth to solve the problem 
optimally. We exemplify our approach on {\sc $p$-Feedback Vertex Set}. In this problem we are given a graph $G$ 
and a positive integer $k$ and the question is 
to check whether there is a subset $F\subseteq V(G)$, $|F|\leq k$, such that $G\setminus F$ is acyclic. The set $F$ is called feedback vertex set of $G$. 

\begin{theorem}
\label{lem:subexpfvsunitmap}
{\sc $p$-Feedback Vertex Set}  admits a parameterized subexponential time algorithm on unit disk graphs and map graphs. 
\end{theorem}
\begin{proof}
We give a subexponential time parameterized  
algorithm on map graphs. An algorithm on unit disk graphs is similar. The algorithm proceeds as follows. Given an instance $(G,k)$, it finds a maximum clique 
$C$ of $G$. Recall that we can find a maximum clique in unit disk graphs and map graphs in polynomial time~\cite{ChenGP98,ChenGP02,ClarkCJ90,RaghavanS03}. 
If $|C|>k+2$, then we return that $G$ does not have feedback vertex set of size at most $k$. 
Next we check whether  $|C|\leq k^\epsilon$ ($\epsilon$ to be fixed later).  If yes then by Theorem~\ref{thm:tstmapgraphs} we know that 
$\tw(G)\leq O(k^{0.5+\epsilon})$. In this case we apply the known algorithm for {\sc Feedback Vertex Set}, that given a tree 
decomposition of width $t$ of a graph $G$ on $n$ vertices, 
finds a minimum sized feedback vertex set in time $2^{O(t\log t)}n^{O(1)}$. Hence in this case 
the running time of our algorithm will be $2^{O(k^{0.5+\epsilon}\log k)} n^{O(1)}$. Now we consider the case when $|C|> k^{\epsilon}$. 
We know that for any feedback vertex set $F$ of $G$,  we have that $|F\cap C|\geq |C|-2$. So we guess the intersection $X=F\cap C$ 
and recursively solve the problem on $(G\setminus X,k-|X|)$. If for any guess we have an yes answer we return yes, else,  
we return no. The running time of this step is guided by the following recurrence 
$T(k)\leq {|C| \choose 2}\cdot T(k-(|C|-2))+|C| \cdot T(k-(|C|-1))+T(k-|C|)$, where the terms ${|C| \choose 2}\cdot T(k-(|C|-2))$, 
$|C|\cdot T(k-(|C|-1))$, $T(k-|C|)$ correspond to choosing $|C|-2$ vertices in 
$F$ from $C$, $|C|-1$ vertices in $F$ from $C$ and $|C|$ vertices in $F$ from $C$, respectively. Roughly,   
$T(k)\leq 3|C|^2 \cdot T(k-|C|)$. This asymptotically solves to $(3|C|)^{2k/|C|}$, which is equal to $2^{O(\frac{2k \log |C|}{|C|})}\leq 2^{O(\frac{2k \log k}{|C|})}$. Hence as $|C|$ increases we have that the function  $2^{O(\frac{2k \log k}{|C|})}$ decreases. Thus the worst case running time is achieved when  $|C|=k^\epsilon$ and hence
this is equal to $ 2^{O(k^{1-\epsilon} \log k)}$. Now we choose $\epsilon$ in  a way that the running time for branching on clique is same as when we 
run a dynamic programming algorithm on graphs of bounded treewidth. Thus we choose an $\epsilon$ such that $2^{O(k^{1-\epsilon} \log k)}=
2^{O(k^{0.5+\epsilon})}$. This gives us that $\epsilon=1/4$ is asymptotically best possible. Thus our algorithm runs in time 
$2^{O(k^{0.75}\log k)} n^{O(1)}=2^{o(k)} n^{O(1)}$. 
\end{proof}

Next we show  that in fact {\sc $p$-(Connected) Vertex Cover}  admits a parameterized subexponential time algorithms 
on Unit Ball Graphs in $\mathbb{R}^d$.

\begin{theorem}
\label{lem:subexpvccvcalld}
 {\sc $p$-Connected Vertex Cover} and {\sc $p$-Vertex Cover} admit a parameterized subexponential time algorithm 
on unit ball graphs of fixed dimension. 
\end{theorem}
 \begin{proof}
Our algorithm for {\sc $p$-(Connected) Vertex Cover} follows along the same line as for {\sc $p$-Feedback Vertex Set}. 
We only outline an algorithm for {\sc $p$-Connected Vertex Cover} here.  

The algorithm proceeds as follows. Given an instance $(G,k)$,  we first check whether 
$k\geq |V(G)|/f(d)$, where $f(d)=2(2^{0.401d(1+o(1))}+1)$. By Lemma~\ref{lem:boundingsolutionsize} we know that 
if $k< |V(G)|/f(d)$, then there is no connected vertex cover of size at most $k$ and hence the answer is no. 
Else we have that $|V(G)|=O(k)$. To implement our algorithm we need a slight 
generalization of problem. We keep a triple $(G',k,X)$ for this problem, where $G'$ 
is the current graph and the objective is to find a set $F\subseteq V(G')$ such that $|F|\leq k$, $F$ is a vertex cover of $G'$ 
and $G[X\cup F]$ is a connected vertex cover of $G$. Essentially the graph $G'$ will be obtained after branching on cliques and the set $X$ will store the 
partially constructed solution so far. This allows us to check connectedness in the whole graph $G$. 
Now the algorithm finds a maximum clique $C$ of $G'$. 
If $|C|>k+1$, then we return that $G'$ does not have a desired set $F$ of size at most $k$. 
Next we check whether  $|C|\leq k^\epsilon $ ($\epsilon$ to be fixed later). 

We first consider the case when $|C|> k^{\epsilon}$.  
We know that for any vertex cover $F$ of $G'$,  we have that $|F\cap C|\geq |C|-1$. So we guess the intersection 
$Z=F\cap C$ and recursively solve the problem on $(G'\setminus Z,k-|Z|,X\cup Z)$. If for any guess we have an 
yes answer we return yes else we return no. The running time of this step is guided by the following recurrence 
$T(k)\leq |C| \cdot T(k-(|C|-1))+T(k-|C|)$,  where the terms  
$|C|\cdot T(k-(|C|-1))$, $T(k-|C|)$ correspond to choosing $|C|-1$ vertices in $F$ from $C$ and $|C|$ vertices in $F$ from $C$, respectively. 
Roughly $T(k)\leq (2|C|) T(k-|C|)$. This asymptotically solves  to $(2|C|)^{k/|C|}$, which is equal to $2^{O(\frac{k \log |C|}{|C|})}\leq 2^{O(\frac{k \log k}{|C|})}$. Hence as $|C|$ increases we have that the function  $2^{O(\frac{2k \log k}{|C|})}$ decreases. Thus the worst case running time is 
achieved when  $|C|=k^\epsilon$ and hence 
this is equal to $ 2^{O(k^{1-\epsilon} \log k)}$.

In the other case we have that  $|C|\leq  k^{\epsilon}$. As discussed before Lemma~\ref{lem:alldvccvc}, by using result 
from~\cite{MillerTTV97} we have that the treewidth of $G'$ is $c_dk^{\epsilon/d}|V(G)|^{1-\frac{1}{d}}=O(k^{1-(1-\epsilon)\frac{1}{d}})$, 
where $c_d$ is a constant depending only on $d$.  In this case we apply a modification of known algorithm for 
{\sc Connected Vertex Cover}, that given a tree 
decomposition of width $t$ of a graph $G^*$ on $n$ vertices, 
finds a minimum sized connected vertex cover in time $2^{O(t\log t)}n^{O(1)}$~\cite{Moser2005}. To solve our problem we do as follows. We first upper 
bound the number of connected components, $\eta_X$, in $G[X]$ by $k^{1-\epsilon}$. Recall that $X$ has been constructed by branching on cliques of 
size at least $k^\epsilon +1$ and thus from each such clique we have at least $k^\epsilon$ vertices in $X$ and vertices from one clique are in one 
component. Thus  $\eta_X \leq k/k^{\epsilon}=k^{1-\epsilon}$. Now we construct a graph $G^*$ as follows. Consider the 
graph $G[X \cup V(G')]$ and contract every connected component in $G[X]$ to a single vertex. Now in the graph $G^*$ the objective is to find a connected 
vertex cover of size at most $k+\eta_X$ such that it contains all the vertices corresponding to connected components in $G[X]$. Now the 
$\tw(G^*)\leq \tw(G')+\eta_X\leq O(k^{1-(1-\epsilon)\frac{1}{d}}+k^{1-\epsilon})$.  Hence in this case the running time of our algorithm is  
$2^{O((k^{1-(1-\epsilon)\frac{1}{d}}+k^{1-\epsilon})\log k)} n^{O(1)}$. 

Now we choose $\epsilon$ in  a way that the running time for branching on clique is same as when we 
run a dynamic programming algorithm on graphs of bounded treewidth. Thus, we choose an $\epsilon$ such that $2^{O(k^{1-\epsilon} \log k)}=
2^{O((k^{1-(1-\epsilon)\frac{1}{d}}+k^{1-\epsilon})\log k)}$. This gives us  that $\epsilon=1/(d+1)$ is asymptotically best possible. 
Thus, our algorithm runs in time $2^{O((k^{1-(1-\epsilon)\frac{1}{d}}+k^{1-\epsilon})\log k)} n^{O(1)}=2^{o(k)} n^{O(1)}$ for every fixed $d$. 
This gives us the desired result. 
\end{proof}

\paragraph{Tractability Borders.}
It is natural to ask how far our approach can be generalized, and in particular, whether many of the  problems discussed so far have EPTASs and parameterized subexponential time algorithms on unit ball graphs in dimension higher than two. In this section we show that one should not expect equally general results for unit ball graphs of dimension at least three. In particular, we show that {\sc Feedback Vertex Set} on Unit Ball Graphs in $\mathbb{R}^3$ does not have an EPTAS unless $P=NP$, and that the problem does not admit a subexponential time parameterized algorithm under the Exponential Time Hypothesis of Impagliazzo, Paturi and Zane~\cite{ImpagliazzoPZ01}. 

\begin{theorem}
\label{thm:FVShard}
{\sc Feedback Vertex Set} on unit ball graphs in $\mathbb{R}^3$ does not admit a PTAS unless $P=NP$, and has no subexponential time parameterized algorithm unless the Exponential Time Hypothesis fails.
\end{theorem}

A {\em unit ball} model of $H$ in $\mathbb{R}^d$ is a map $f : V(H) \rightarrow \mathbb{R}^d$ such that $u$ and $v$ are adjacent iff the euclidean distance between $f(u)$ and $f(v)$ is at most $1$. In the construction it is much more convenient 
to work with this alternate definition of unit ball graphs rather than saying that $f(u)$ and $f(v)$ is at most $2$ 
and hence we use this alternate definition in this section. 
In our constructions no two vertices will map to the same point, and thus we will often refer to vertices in $H$ by the points in  
$\mathbb{R}^d$ which they map to. For the proof of Theorem~\ref{thm:FVShard} we need the following lemmas. It appears that the following lemma can easily be derived from the results in~\cite{Eades:2000fk} about the three dimensional orthogonal graph drawings. However, since we could not find this result explicitly, we  give a proof here for completeness. 

\begin{lemma}
\label{lem:construct} For any graph $G$ on $n$ vertices of maximum degree $6$, there is a unit ball graph $H$ on $O(n^2)$ vertices such that $H$ is a subdivision of $G$. Furthermore, $H$ and a unit ball model of $H$ in $\mathbb{R}^3$ can be constructed from $G$ in polynomial time. \end{lemma}

The proof of Lemma~\ref{lem:construct} is straightforward, but somewhat tedious. 

\begin{proof}
In this construction we envision the $x$-axis as being horizontal with positive direction towards the right, the $z$-axis being vertical with positive direction upwards. The intuition behind the proof is that every vertex of $G$ is assigned its own ``fat'' $x$-$z$ plane. The edges of $G$ are routed parallel to the $y$ axis in the $y-x$ plane with $z=0$, and in each ``fat'' $x$-$z$ plane we ensure that the edges connect to their corresponding vertex. This local routing of edge endpoints to a vertex happens above the $y-x$ plane with $z=0$ and does not interfere with the global routing of the edges.

For a point with integer coordinates $(x,y,z)$ and integer $\ell$ define the set $L[x,y,z]_{\bf x}^\ell$ to be $\{(x+x',y,z) : |x'|+|\ell-x'|=|\ell|\}$. In particular, if $\ell$ is positive then $L[x,y,z]_{\bf x}^\ell$ contains $\{(x,y,z),(x+1,y,z),(x+2,y,z),\ldots,(x+\ell,y,z)\}$, while if $\ell$ is negative then  $L[x,y,z]_{\bf x}^\ell$ contains $\{(x,y,z),(x-1,y,z),(x-2,y,z),\ldots,(x-\ell,y,z)\}$. Similarly we define $L[x,y,z]_{\bf y}^\ell$ to be $\{(x,y+y',z) : |y'|+|\ell-y'|=|\ell|\}$ and $L[x,y,z]_{\bf z}^\ell$ to be $\{(x,y,z+z') : |z'|+|\ell-z'|=\ell\}$. Given three integers $x$,$y$,$z$, the graph $P[x,y,z]$ corresponds to the point set
\begin{eqnarray*}
P[x,y,z] & = & L[x,y,z]_{\bf z}^{-2} \cup L[x,y,z-2]_{\bf x}^{2} \\
& \cup & L[x,y,z]_{\bf x}^{2} \\
& \cup & L[x,y,z]_{\bf z}^{2} \cup L[x,y,z+2]_{\bf x}^{2} \\
& \cup & L[x,y,z]_{\bf x}^{-2} \cup L[x-2,y,z]_{\bf z}^{4} \cup L[x-2,y,z+4]_{\bf x}^{4} \\
& \cup & L[x,y,z]_{\bf y}^{-2}  \cup L[x,y-2,z]_{\bf z}^{6} \cup L[x,y-2,z+6]_{\bf y}^{2} \cup L[x,y,z+6]_{\bf x}^2 \\
& \cup & L[x,y,z]_{\bf y}^{2}  \cup L[x,y+2,z]_{\bf z}^{8} \cup L[x,y+2,z+8]_{\bf y}^{-2} \cup L[x,y,z+8]_{\bf x}^2
\end{eqnarray*}
The set $P[x,y,z]$ corresponds to a vertex of degree $6$ in $[x,y,z]$, and there are $6$ paths, each starting in $(x,y,z)$ and ending in $(x+2,y,z-2)$, $(x+2,y,z)$, $(x+2,y,z+2)$, $(x+2,y,z+4)$, $(x+2,y,z+6)$ and $(x+2,y,z+8)$ respectively. The $y$-coordinate of any intermediate point on the paths is always between $y-2$ and $y+2$. Any points that are generated twice still correspond only to one single vertex.

For an integer $y$ and six integers $x_1 < x_2 < \ldots < x_6$ such that $x_{i+1} - x_i \geq 2$ we define ${\bf P}[y,x_1,x_2,x_3,x_4,x_5,x_6]$ to be the point set 
\begin{eqnarray*}
{\bf P}[y,x_1,x_2,x_3,x_4,x_5,x_6] & = & P[-2,y,12] \\
& \cup & \bigcup_{i=1}^6 L[0,y,10+2(i-1)]_{\bf x}^{x_i} \cup L[x_i,y,10+2(i-1)]_{\bf z}^{-10-2(i-1)}
\end{eqnarray*}
The set ${\bf P}[y,x_1,x_2,x_3,x_4,x_5,x_6]$ corresponds to a vertex of degree $6$ in $[-2,y,12]$ with $6$ paths starting in this vertex end ending in $[x_i,y,0]$ for $1 \leq i \leq 6$. The $y$-coordinate of the intermediate vertices on the path is between $y-2$ and $y+2$. In this sense, the paths corresponding to the vertex in $[-2,y,12]$ are routed in a ``fat'' $x-z $-plane.

We are now ready to construct $H$ given $G$. We give the construction for $6$-regular graphs $G$ and then explain how to modify the construction to the case when $G$ has maximum degree $6$. We label the vertices in $G$ by $v_1, \ldots v_n$ and the edges of $G$ by $e_1, \ldots e_m$ with $m \leq 3n$. For every $i \leq m$ define $a(i)$ and $b(i)$ such that the endpoints of the edge $e_i$ are $v_{a(i)}$ and $v_{(b(i))}$ respectively. Now, for every vertex $v_i$ let $x^i_1 < x^i_2 < \ldots < x^i_6$ be integers so that $v_i$ is incident to the edges $e_{x^i_j}$ for $1\leq j \leq 6$. For every vertex $v_i$ we add the point set ${\bf P}[10i,2x^i_1,2x^i_2,2x^i_3,2x^i_4,2x^i_5,2x^i_6]$. Finally for every edge $e_i$ we add the set $L[2i,10a(i),0]_{\bf y}^{10(b(i)-a(i))}$. This concludes the construction of $H$. 

It is easy to see that $H$ can be constructed from $G$ in polynomial time. Furthermore, it is easy to verify that $H$ has $O(n^2)$ vertices since $m \leq 3n$. To see that $H$ is a subdivision of $G$ observe that when $G$ has an edge $e_t$ between $v_i$ and $v_j$, in $H$ there is a path from the point $[-2,10i,12]$ through $[2t,10i,0]$ and $[2t,10j,0]$ to the point  $[-2,10j,12]$. This concludes the proof of the lemma.
\end{proof}
\begin{lemma}
\label{lem:constructFVS} There is a polynomial time algorithm that given a graph $G$ on $n$ vertices of maximum degree $3$ outputs a unit ball graph $H$ together with a unit ball model of $H$ in $\mathbb{R}^3$, such that given any vertex cover $C$ of $G$, a feedback vertex set $S$ of $H$ of size at most $|C|$ can be computed in polynomial time, and given any feedback vertex set $S$ of $H$, a vertex cover $C$ of $G$ of size at most $|S|$ can be computed in polynomial time. 
\end{lemma}
\begin{proof}
Given $G$ we start by applying the well-known construction for transforming instances of {\sc Vertex Cover} to instances of {\sc Feedback Vertex Set}. We construct $G'$ from $G$ by adding a vertex $x_{uv}$ for every edge $uv$ of $G$ and making $x_{uv}$ adjacent to $u$ and to $v$. Since the maximum degree of $G$ was $3$, the maximum degree of $G'$ is $6$. Now we apply Lemma~\ref{lem:construct} to $G'$ obtain the graph $H$ and a unit ball model of $H$. Every vertex cover $C$ of $G$ is a feedback vertex set of $G'$, and since $H$ is a subdivision of $G'$, every vertex cover of $G$ is a feedback vertex set of $H$. For the reverse direction, it is well-known that given a feedback vertex set $S$ in a graph, one can find in polynomial time a feedback vertex set $S'$ of size at most $|S|$ such that all vertices in $S'$ have degree at least $3$~\cite{BarYGJ98}. Let $S'$ be a feedback vertex set of $H$ such that every vertex in $S'$ has degree at least $3$ in $H$. Then, every vertex in $S'$ is also a vertex in $G$. We claim that $S'$ is a vertex cover of $G$. Let $uv$ be an edge in $G$. Then $u$,$x_{uv}$,$v$ is a cycle in $G'$ and since $H$ is a subdivision of $G'$, $H$ contains a cycle going through $u$,$x_{uv}$ and $v$ where all vertices in the cycle except $u$ and $v$ have degree at most $2$. Since $S'$ is a feedback vertex set of $H'$ containing no vertices of degree less than $3$, $S'$ contains either $u$ or $v$. Hence $S'$ is a vertex cover of $G$.
\end{proof}

If a subexponential time parameterized algorithm for {\sc Feedback Vertex Set} on unit ball graphs in $\mathbb{R}^3$ existed we could combine it with Lemma~\ref{lem:constructFVS} to get a subexponential time algorithm for {\sc Vertex Cover} on graphs of maximum degree $3$. Similarly, a PTAS for {\sc Feedback Vertex Set} on unit ball graphs in $\mathbb{R}^3$ could be combined with Lemma~\ref{lem:constructFVS} to yield a PTAS for {\sc Vertex Cover} on graphs of maximum degree $3$. Since {\sc Vertex Cover} is known not to admit a $(1+\epsilon)$-factor approximation algorithm, for some fixed $\epsilon>0$, on graphs of degree at most $3$ unless $P=NP$~\cite{AlimontiK00}, and not to have subexponential time parameterized algorithms on graphs of degree at most $3$ under the Exponential Time Hypothesis~\cite{ImpagliazzoPZ01}, we obtain 
Theorem~\ref{thm:FVShard}.

 \newcommand{\bibremark}[1]{\marginpar{\tiny\bf#1}}
  \newcommand{\biburl}[1]{\url{#1}}

\end{document}